\documentclass[a4paper]{article}

\usepackage[utf8]{inputenc}
\usepackage[T1]{fontenc}
\usepackage{textcomp}
\usepackage[british]{babel}
\usepackage[left=1in,right=1in,top=1in,bottom=0.5in,includefoot=true]{geometry}
\usepackage{amsmath, amsfonts, amsthm, amssymb}
\usepackage[shortlabels]{enumitem}
\usepackage{mathrsfs}
\usepackage{bbm}
\usepackage{accents}
\usepackage[backend=biber,doi=false,isbn=false,maxbibnames=99]{biblatex}
\usepackage{csquotes}
\usepackage{microtype}
\usepackage{xcolor}
\usepackage{booktabs}
\usepackage{graphicx}
\usepackage{zref-clever}
\usepackage{hyperref}

\newcommand{\cref}[1]{\zcref{#1}}
\newcommand{\Cref}[1]{\zcref[S]{#1}}
\zcsetup{cap,nameinlink=false}
\zcRefTypeSetup{equation}{abbrev}
\AddToHook{env/definition/begin}{%
	\zcsetup{countertype={definition=definition}}}
\AddToHook{env/assumption/begin}{%
	\zcsetup{countertype={definition=assumption}}}
\AddToHook{env/example/begin}{%
	\zcsetup{countertype={definition=example}}}
\AddToHook{env/theorem/begin}{%
	\zcsetup{countertype={definition=theorem}}}
\AddToHook{env/corollary/begin}{%
	\zcsetup{countertype={definition=corollary}}}
\AddToHook{env/proposition/begin}{%
	\zcsetup{countertype={definition=proposition}}}
\AddToHook{env/lemma/begin}{%
	\zcsetup{countertype={definition=lemma}}}
\AddToHook{env/remark/begin}{%
	\zcsetup{countertype={definition=remark}}}

\DeclareFieldFormat{eprint:ssrn}{%
  SSRN\addcolon\space
  \ifhyperref
    {\href{https://ssrn.com/abstract=#1}{\nolinkurl{#1}}}
    {\nolinkurl{#1}}%
}
\DeclareFieldAlias{eprint:SSRN}{eprint:ssrn}

\newcommand{\N}{\mathbb{N}}
\newcommand{\R}{\mathbb{R}}
\newcommand{\E}{\mathbb{E}}
\newcommand{\Q}{\mathbb{Q}}

\makeatletter \renewcommand\d[1]{\ensuremath{%
  \;\mathrm{d}#1\@ifnextchar\d{\!}{}}}
\makeatother

\newcommand{\condbar}{\;\middle|\;}
\newcommand{\ubar}[1]{\underaccent{\bar}{#1}}
\newcommand{\one}{\textbf{1}}

\theoremstyle{definition}
\newtheorem{definition}{Definition}[section]
\newtheorem{assumption}[definition]{Assumption}
\newtheorem{example}[definition]{Example}
\theoremstyle{plain}
\newtheorem{theorem}[definition]{Theorem}
\newtheorem{corollary}[definition]{Corollary}
\newtheorem{proposition}[definition]{Proposition}
\newtheorem{lemma}[definition]{Lemma}
\theoremstyle{remark}
\newtheorem{remark}[definition]{Remark}

\zcRefTypeSetup{assumption}{
Name-sg = Assumption ,
name-sg = assumption ,
Name-pl = Assumptions ,
name-pl = assumptions ,
}

\numberwithin{table}{section}
\numberwithin{figure}{section}

\numberwithin{equation}{section}

\definecolor{linkcolor}{cmyk}{.2,.98,1.,.18}
\definecolor{citecolor}{cmyk}{.49,.0,.85,.38}
\hypersetup{colorlinks=true,linkcolor=linkcolor,citecolor=citecolor}

\allowdisplaybreaks[1]

\usepackage{authblk}

\title{Optimal Investment and Consumption in Financial Markets with Integrated Variance Clocks}

\author[1]{Eduardo Abi Jaber}
\author[2]{Florian Gutekunst}
\author[2,3]{Martin Herdegen}
\author[2]{David Hobson}
\affil[1]{CMAP, École Polytechnique}
\affil[2]{Department of Statistics, University of Warwick}
\affil[3]{Institut für Stochastik und Anwendungen, University of Stuttgart}

\date{22nd September 2026}

\makeatletter
\def\blfootnote{\xdef\@thefnmark{}\@footnotetext}
\makeatother

\begin{document}
		\maketitle
		\blfootnote{Eduardo Abi Jaber is grateful for the financial support from the Chaires FiME-FDD and  Financial Risks at Ecole Polytechnique.  Florian Gutekunst gratefully acknowledges funding from the Statistics Centre for Doctoral Training at the University of Warwick.}

		\begin{abstract}
                We  study the infinite-horizon optimal investment and consumption problem in a general class of continuous financial markets, where uncertainty is driven by a continuous non-decreasing stochastic clock representing accumulated variance. This framework encompasses classical Markovian and non-Markovian stochastic volatility models as well as singular realized-variance models in which no spot volatility process exists. We characterize the value process and optimal investment and consumption strategies in terms of a non-linear infinite-horizon backward stochastic differential equation driven jointly by calendar time and the stochastic clock. We develop a general well-posedness theory for this new class of IVC-BSDEs based on the method of sub- and supersolutions, establishing existence, uniqueness, and stability under natural conditions that might be of independent interest beyond the financial application at hand. We are moreover able to identify the sign of the $Z$-component of the solution using Malliavin calculus. We then apply our results to Volterra Heston models with locally integrable kernels, covering both rough and hyper-rough regimes. Exploiting the affine structure of the model, we verify the optimality of the candidate strategies in incomplete markets and obtain an explicit representation of the solution in the complete market case. Owing to the generality of the framework and the weak assumptions imposed on the stochastic clock, our results unify and extend several existing results for optimal investment and consumption, including  classical Markovian stochastic volatility models.
		\end{abstract}

\section{Introduction}

Continuous-time portfolio optimization traditionally models uncertainty through an instantaneous volatility process. In this paradigm, the asset price evolves in calendar time, while its local variability is described by a spot volatility process $\sigma$, so that the continuous quadratic variation of the logarithm of the asset price is given by $\int_0^{\cdot} \sigma_t^2dt.$ 
This viewpoint underlies virtually all continuous-time models for optimal portfolio allocation, ranging from classical Markovian stochastic volatility models (\textcite{Zariphopoulou2001,  lim2004quadratic, kraft2005optimal, liu2007portfolio, Gutekunst2025}) to more recent non-Markovian (and rough) Volterra volatility models (\textcite{AbiJaber2021a, Aichinger2023, Baeuerle2020,  dupret2021portfolio, Han2020, Han2021}). Even in these more general settings, the primitive object is the spot volatility process, from which accumulated variance is subsequently obtained by integration.

This modeling philosophy is largely driven by mathematical tractability. Indeed, spot volatility is not directly observable and is notoriously difficult to estimate accurately from market data. By contrast, cumulative quantities such as realized variance are more easily inferred from directly observable quantities and statistically more robust. This suggests that it may be more natural to take accumulated variance “$\int_0^{\cdot} \sigma^2_t dt$” itself as the primitive source of randomness and to view asset prices as evolving in stochastic market time rather than deterministic calendar time. This observation motivates the central modelling principle of this paper: rather than specifying a spot volatility process, we take integrated variance itself as the primitive state variable. Such perspective has been formalized by \Textcite{clark1973subordinated} through subordinated Brownian models, where returns are indexed by an increasing stochastic clock representing the flow of market activity. In mathematical finance, this viewpoint evolved into the theory of time-changed Lévy processes (\Textcite{carr2004time}), see also \Textcite{mendoza2010time}.

Motivated by this approach, we study optimal investment and consumption in a generic continuous class of models within the market time paradigm, where the stock price evolves according to a Brownian motion $W$ under a general continuous non-decreasing stochastic clock $U$. More precisely, under the physical measure $\mathbb P$, the random fluctuations of the stock price are driven by the time-changed Brownian motion $W_U$, while the risk premium is assumed to accumulate absolutely continuously with respect to the stochastic clock through a term of the form $
\int_0^{\cdot}\Lambda_sdU_s$,
for some predictable process $\Lambda$. This reflects the idea that excess returns scale with realized variance rather than calendar time. The stock price is therefore modeled as
\begin{equation*}
\frac{dS_t}{S_t} = r_t dt + 
\Lambda_t dU_t + \rho dW_{U_t} + \sqrt{1-\rho^2} dW^\perp_{U_t},
\end{equation*}
where $r$ denotes the short-rate process. Unlike the classical independent subordination framework, the stochastic clock $U$ is allowed to depend on the Brownian motion $W$ (but not $W^\perp$), creating a dependence between the asset price and its realized variance, the so-called leverage effect. Such a formulation naturally encodes market incompleteness whenever $|\rho| \neq 1$.

This framework contains the classical stochastic volatility setting as a particular case. Indeed, whenever the clock is absolutely continuous with respect to Lebesgue measure, one may write $
U_t=\int_0^t \sigma_s^2 ds$
for some process $\sigma$, in which case the time-changed Brownian motion admits the representation
$dW_{U_t}=\sigma_t dB_t$ for some Brownian motion $B$, thereby recovering the usual stochastic volatility form. Typically, the volatility process $\sigma$ is specified as a function of an underlying factor $Y$, that is $\sigma_t=\sigma(Y_t)$, where $Y$ follows either Markovian or non-Markovian dynamics. 

A prototypical Markovian example is when $Y$ is given by an autonomous diffusion driven by a Brownian motion $B$.
A prototypical non-Markovian example is given by stochastic Volterra models of the form
\begin{equation}\label{eq:introvolterra}
Y_t
= Y_0 + \int_0^t K(t-s)\mu(Y_s)ds + \int_0^t K(t-s)\eta(Y_s)dB_s.
\end{equation}
with kernels $K\in L^2$, where again $B$ is a Brownian motion.
In this setting the factor process $Y$ may fail to be Markovian and may even fail to be a semimartingale, for instance, for the fractional kernel $K(t)=t^{H-1/2}$, with $H \in (0,1)$ and $H \neq 0.5$. This places the corresponding portfolio optimization problem outside the scope of classical dynamic programming methods and in this Volterra setting, optimal investment and consumption problems remain largely unexplored. 

We emphasise that our framework not only encompasses the Markovian and non-Markovian Volterra settings described above, it also goes beyond them by allowing for a general clock $U$, which need not be absolutely continuous with respect to Lebesgue measure. In particular, we allow for singular realized-variance dynamics, where there may be no process $\sigma$ such that $U_t=\int_0^t \sigma_s^2ds$. Financially, this makes realized variance the primitive object of the model, rather than an instantaneous volatility process, which is typically unobservable and difficult to estimate. Mathematically, such singular clocks naturally arise in highly irregular volatility models, including hyper-rough Volterra models with kernels in $L^1$ but not in $L^2$ as studied by \Textcite{AbiJaber2021, Jusselin2020}. From this perspective, continuous but singular clocks can be viewed as an intermediate regime between classical stochastic volatility models, where the clock is absolutely continuous, and jump-type time changes, where the clock itself is discontinuous. In fact, certain  continuous Volterra clocks converge to pure-jump processes, when the kernel becomes sufficiently singular and falls outside $L^1$, see for instance \Textcite{abijaber2026volterra}.

The singular and non-Markovian nature of the clock creates serious challenges for optimal investment and consumption. Classical Hamilton--Jacobi--Bellman methods are not available without a Markovian state variable. Standard semimartingale techniques are also difficult to apply directly, since the natural control and utility accumulation may take place with respect to the random measure $dU$ rather than the Lebesgue measure. The purpose of this paper is to develop a tractable approach to optimal investment and consumption in this general incomplete time-changed market. Our method is based on a new class of non-linear infinite-horizon backward stochastic differential equations formulated directly with respect to the integrated variance clock
\begin{equation}\label{eq:introIVCBSDE}
dF_s = (a_s F_s + g(F_s)) ds + b_s F_s dU_s + Z_s dW_{U_s}
\end{equation}
where $a, b$ are some processes, $g$ is an increasing (but not necessarily Lipschitz) function, and the finite-variation measure $dU$ may be singular with respect to Lebesgue measure.

\subsection*{Our contributions}

Our first contribution is to develop a general theory for solving a new class of non-linear infinite-horizon BSDEs.
Since the driver of the BSDE is non-Lipschitz in the $Y$-component, standard existence results based on the contraction mapping theorem do not apply here.
However, by exploiting the monotonicity of the driver, we develop a theory of sub- and supersolutions.
We prove that a solution always lies between a pair of ordered sub- and supersolutions, and we establish the existence of maximal and minimal solutions, which are instrumental for establishing uniqueness.
Under a uniform lower boundedness condition, we obtain existence of a solution as well as uniqueness (Theorem~\ref{existence:uniformly_wellposed}) and stability (Theorem~\ref{stability:uniformly_wellposed}).
Moreover, under an additional uniform upper boundedness condition, we provide a necessary and sufficient criterion for the existence of a solution (Theorem~\ref{existence:eta_bounded}).
Under a sign condition on the Malliavin derivative of one of the driving processes in the BSDE, we are able to identify the sign of the $Z$-component, which seems to be a novel contribution to BSDE theory (\cref{properties:volatility:non-positive}).

Our second contribution is to connect the solution of the BSDE to the value process and optimal controls of the infinite-horizon investment-consumption problem. 
This yields a generic framework for optimal consumption and investment under general realized-variance dynamics, including singular, non-Markovian, and non-semimartingale volatility structures.
We prove a general verification result (Theorem~\ref{verification:verification}) relying on a duality argument, extending classical methods to our non-Markovian time-changed setting.
As an immediate application, this yields verification in strongly myopically well-posed models with bounded myopic consumption rate (\cref{verification:example:bounded_eta}).
Here, the myopic consumption rate is the optimal consumption rate in the Black-Scholes model that arises from holding the stochastic factor constant in its current state, and the model is called strongly myopically well-posed if the myopic consumption rate is positive and bounded away from zero.

Our third contribution in Section~\ref{section:heston} is to demonstrate that the general integrated variance clock framework applies to Volterra Heston models with locally integrable kernels as in \textcite{AbiJaber2021}.
This encompasses both locally square-integrable kernels and hyper-rough regimes with locally integrable kernels, where the integrated variance is not absolutely continuous and the spot variance does not exist (see \cite{Jusselin2020}). 
We show that our verification theorem applies in both complete and incomplete market settings, using different strategies for each case. The key technical step is to establish the martingale property of a certain Doléans--Dade exponential for which the standard Novikov criterion does not apply. In the complete setting, we show that the BSDE admits an explicit representation in terms of a Riccati-Volterra equation, allowing us to establish the required martingale property using relatively standard arguments from the theory of affine Volterra processes (Theorem~\ref{heston:complete:verification}). In the incomplete setting, no such explicit representation is available. Instead we derive tight upper and lower bounds for the BSDE solution and prove that its $Z$-component is non-positive. 
This latter property provides a key ingredient that enables us to establish the required martingale property, and the analysis of a BSDE via study of the sign of the $Z$-component is 
a novel contribution to the BSDE literature in its own right (Theorem~\ref{heston:verification}).

\subsection*{Related literature}
        Optimal investment and consumption is a long-studied classical problem in mathematical finance.
        In a Black-Scholes market, i.e. a model with constant coefficients, this was first studied by \textcite{Merton1969,Merton1971} and can be treated by both primal and dual methods, see e.g. \textcite{Karatzas1986,Davis1990,Herdegen2021}.

        The situation becomes more difficult when one introduces an additional stochastic factor that drives the coefficients of the asset price.
		Explicit solutions to the infinite-horizon inverstment-consumption problem are then only known when the market is complete or under logarithmic utility; in general the optimal controls are given in terms of the solution to some non-linear equation.
		Even once the existence of a solution to this equation is shown, the lack of an explicit expression makes it hard to verify the properties of the solution that are needed to give a rigorous verification argument.

		Classically, it is assumed that the stochastic factor can be represented as an Itô diffusion.
        Over the finite horizon, the problem without intermediate consumption in an incomplete market was studied by \textcite{Karatzas1991,Zariphopoulou2001}. 
		The consumption problem in a complete market was analysed by \textcite{Wachter2002}.
		Over the infinite-horizon, \textcite{Fleming2003} considered the consumption problem in an incomplete market in a specific stochastic factor model.
		\Textcite{Hata2012,Hata2012a} study the problem under uniform ellipticity constraints. 
		Importantly, none of the above papers' settings includes a Heston-type model.
        Recently, \Textcite{Ferrari2026} investigated the related problem of optimal consumption in the presence of labour income in the Kim-Omberg model in a complete market.

        \Textcite{Guasoni2020} study the incomplete-market consumption problem in a very general diffusion setting.
		Using the method of sub- and supersolutions, they provide existence results for the candidate solution and prove a very general verification result.
		They use their results to prove the existence of a candidate solution in the Heston model, but do not verify the assumptions of the verification theorem in this model.
		In a follow-up paper, \textcite{Guasoni2019} consider the Vasicek stochastic interest rate model. 
		They show the existence of a candidate solution using the method of sub- and supersolutions, and prove optimality using a specifically tailored verification theorem.
		Through a variational formulation, \textcite{Guasoni2025} are able to prove existence and verification for strongly myopically well-posed models if the stochastic factor is ergodic and its scale function diverges at the boundary of the state space. 
		In the complete market setting, their results apply for general relative risk aversion $R$, but are restricted to $R \in (0, 1)$ for incomplete markets.
		The variational formulation also leads to a numerical scheme for computing the optimal consumption rate.
		
		Recently, \textcite{Gutekunst2025} developed methods that allow for proving existence and verification in a large class of models.
		Using sub- and supersolutions that are proportional to the myopic consumption rate (i.e. the optimal consumption rate in the Black-Scholes model that arises from holding the stochastic factor constant in its current state), they show the existence of a global positive solution for a large class of models.
		Using the specific form of these sub- and supersolutions, they are then able to characterise the asymptotic behaviour of the solution and its logarithmic derivative.
		This then allows for showing that the candidate controls are indeed optimal.
		In particular, these methods allow them to show existence and verification in the Heston model.
		Moreover, they give a full characterisation of the well-posedness of the investment-consumption problem when the stochastic factor is modelled as a continuous-time Markov chain with a finite number of states.
		In this finite regime setting, the optimal consumption rate can be approximated numerically through a fixed point iteration.
		By discretising, this yields an efficient numerical scheme when the stochastic factor is an Itô diffusion.

        While the aforementioned literature assumes that the stochastic factor can be represented as an Itô diffusion, there has recently been considerable interest in non-Markovian volatility models. In particular, the class of Volterra volatility models as in \eqref{eq:introvolterra} has attracted significant attention because of its flexibility and strong empirical performance. The choice of the kernel allows one to capture a wide range of temporal dependence structures, including fractional (\textcite{comte1998long}) and  rough volatility (\textcite{Gatheral2018}),  path-dependent volatility dynamics (\textcite{guyon2023volatility}), and multifactor specifications (\textcite{AbiJaber2019a}). Extensive empirical studies, such as \textcite{abi2025volatility}, have shown that Volterra volatility models provide  excellent fits to S\&P 500 option prices over a wide range of maturities and market regimes. From a mathematical perspective, however, moving to a non-Markovian setting substantially complicates portfolio optimization, since one can no longer rely on the Hamilton--Jacobi--Bellman equation.

        Despite these challenges, several portfolio optimization problems have been solved over finite horizons mostly within the tractable affine Volterra framework of \textcite{AbiJaber2019}. In particular, the Volterra Heston model introduced for the fractional kernel by  \textcite{ElEuch2019}, has become the canonical example, and has been studied in the context of mean-variance portfolio selection.
        \Textcite{Han2020} study the Markowitz mean-variance portfolio selection problem in a one-asset Volterra Heston market.
		\Textcite{AbiJaber2021a} extend their results to multivariate affine and quadratic Volterra models.
		\Textcite{Fouque2018,Fouque2019} consider the Merton problem in a stochastic factor model in which the stochastic factor is a general stochastic process adapted to the Brownian filtration and the market price of risk is bounded.
		\Textcite{Baeuerle2020} treat the Merton problem in a fractional Heston-type model.
		\Textcite{Han2021} consider the Merton problem in a single-asset Volterra Heston model with locally square-integrable kernel and relative risk aversion $R \in (0, 1)$, which was extended by \textcite{Aichinger2023} to multivariate Volterra-Wishart models (which include multivariate Volterra-Heston models).
		Under Epstein-Zin stochastic differential utility, the investment-consumption problem in a non-Markovian stochastic factor model was treated by \textcite{Feng2026} under an exponential integrability condition on the interest rate and market price of risk.

        Our work is also related to the literature on (backward) stochastic differential equations driven by random clocks of the type \eqref{eq:introIVCBSDE}. 
        Forward SDEs of this type were studied by \textcite{Kobayashi2010}.
        A smaller literature considers the corresponding backward theory, see e.g.\ \textcite{Chen2026,DiNunno2014}.
        The works on BSDEs typically formulate the equation solely w.r.t.\ the stochastic clock.
        By contrast, the investment-consumption problem considered here leads to a BSDE whose finite variation part is driven both by calendar time and the stochastic clock.

        Thus, all of these works consider finite-horizon problems or require additional structural assumptions on the volatility dynamics. The present paper develops a general infinite-horizon theory formulated directly in terms of the integrated variance clock, allowing both absolutely continuous and singular variance clocks within a unified framework.

        \subsection*{Outline}

        The remainder of the paper is organized as follows. \Cref{section:setting} introduces the problem setting.
        \Cref{section:hjb} studies the associated infinite-horizon BSDE, develops the method of sub- and supersolutions, and establishes the existence, uniqueness, stability, and sign identification results. 
        \Cref{section:verification} provides the general verification theorem and its application to models with bounded myopic consumption rate. 
        Finally, \cref{section:heston} applies the theory to the Volterra Heston model. 
        \Cref{appendix:auxiliary} collects some auxiliary technical results.
        
		\section{Integrated Variance Clock  Models}
		\label{section:setting} 
        
        We work on a filtered probability space $(\Omega, \mathcal{F}, \mathbb{F} = (\mathcal{F}_t)_{t\ge 0}, \mathbb{P})$, where the filtration $\mathbb{F} = (\mathcal{F}_t)_{t\ge 0}$ is assumed to be right-continuous and to satisfy the extension property that whenever $(\Q_t)_{t \geq 0}$ is a coherent family of probability measures on $(\Omega, \mathcal{F})$, i.e., $\Q_t$ and $\Q_s$ coincide on $\mathcal{F}_{t \wedge s}$ for all $t, s \geq 0$, then there exists a probability measure $\Q$ on $(\Omega, \mathcal{F})$ such that $\Q$ coincides with $\Q_t$ on $\mathcal{F}_t$ for all $t \geq 0$. Note that this implies that $\mathbb{F}$ cannot satisfy the completeness assumption made in the usual conditions. For this reason, we assume instead that $\mathbb{F}$ satisfies the so-called natural augmentation; see \textcite[Section 2]{najnude2011}. Then
         virtually all results needed for stochastic calculus carry over verbatim but the extension property is preserved; see \textcite[Section 3 and Corollary 4.10]{najnude2011}. Also note that the extension property is always satisfied if we work on the (naturally augmented) path spaces of continuous or c\`adl\`ag functions in $\R^d$ for some $d > 0$; see \textcite[Proposition 4.4]{najnude2011}.

        Throughout the paper, we denote by $\mathbb{F}^{X}$ the right-continuous and naturally augmented filtration generated by an $\mathbb{F}$-adapted process $X$.

        We assume that the interest rate $r$, excess return per unit variance $\Lambda$ and integrated variance process $U$ are all progressively measurable w.r.t.\ $\mathbb{F}$.
        We assume that $U$ is continuous and non-decreasing with $U_0 = 0$, and that \[
                \int_0^t |r_s| \d s + \int_0^t \Lambda_s^2 \d U_s < \infty \text{ a.s.}
        \] for all $t > 0$.
        Note that none of the processes $r, \Lambda, U$ are assumed to be Markovian.
        
        We consider the following model for the bank account $S^0$ and risky asset price process $S$:
       \begin{equation}
                \label{eq:IVC} 
                \begin{aligned}
				        \frac{dS^0_t}{S^0_t} &= r_t dt, \\
				        \frac{dS_t}{S_t} &= r_t dt + \Lambda_t d U_t + \rho dM_t + \sqrt{1-\rho^2} dM^{\perp}_t,
		          \end{aligned} 
        \end{equation}
        where $M$ and $M^\perp$ are strongly orthogonal continuous $\mathbb{F}$-local martingales with quadratic variation $\langle M \rangle = \langle M^\perp \rangle = U$ and $\rho \in [-1, 1]$ is constant.

        By the Dambis–Dubins–Schwarz theorem, possibly after enlarging the probability space, \eqref{eq:IVC} can alternatively be written as 
        \begin{equation}
                \label{eq:IVC_DDS}
                \frac{dS_t}{S_t} = r_t dt + \Lambda_t d U_t + \rho dW_{U_t} + \sqrt{1-\rho^2} dW^{\perp}_{U_t},
        \end{equation}
        where $(W, W^\perp)$ is the Dambis–Dubins–Schwarz Brownian motion associated to $(M, M^\perp)$.
       
        We call models of the form \eqref{eq:IVC}/\eqref{eq:IVC_DDS} integrated variance clock (IVC) models. The key modeling principle is to regard the integrated variance process $U$ as the primitive object, rather than an instantaneous variance process. Classical stochastic volatility models are recovered when $U$ is absolutely continuous with respect to the Lebesgue measure, i.e.~$U=\int_0^{\cdot} \sigma_s^2 ds$ where $\sigma$ is the spot volatility process, while singular clocks naturally accommodate  models for which no spot variance exists. Both these cases are illustrated in the following examples.

        \begin{example}[Stochastic volatility models]
                \label{setting:typical_factor_model}
                A typical stochastic volatility factor model is of the form
        \begin{equation}
                \label{setting:typical_factor_model_dynamics}
                \begin{aligned}
                        \frac{dS^0_t}{S^0_t} &= r(Y_t)\,dt, \\
                        \frac{dS_t}{S_t}
                        &= \bigl(r(Y_t)+\lambda(Y_t)\sigma(Y_t)\bigr)\,dt
                        +\sigma(Y_t)\left(\rho\,dB_t+\sqrt{1-\rho^2}\,dB_t^\perp\right),
                \end{aligned}
        \end{equation}
        where $B$ and $B^\perp$ are independent Brownian motions and the stochastic factor $Y$ is adapted to the filtration generated by $B$. This class of models forms the absolutely continuous subclass of IVC models \eqref{eq:IVC} by choosing $r_t=r(Y_t),\Lambda_t=\frac{\lambda(Y_t)}{\sigma(Y_t)}$, together with
        \[
                U_t=\int_0^t\sigma(Y_s)^2\,ds,\qquad
                M_t=\int_0^t\sigma(Y_s)\,dB_s,\qquad
                M_t^\perp=\int_0^t\sigma(Y_s)\,dB_s^\perp
        .\]
        The factor process $Y$ may be either Markovian or non-Markovian. For example, taking $\sigma(y)=\sqrt{y}, \lambda(y)=\Lambda \sqrt{y}$, yields the affine Heston class: 
        \begin{itemize}
                \item The classical Markovian \textcite{heston1993closed}  model is obtained by letting \[
                    dY_t
                    =
                    -\kappa(Y_t-\theta)\,dt
                    +
                    \nu\sqrt{Y_t}\,dB_t.
                \]
\item  The class of Volterra Heston models of \textcite[Section~7]{AbiJaber2019} is obtained by taking
        \begin{align}\label{eq:exvolterra}        Y_t
                =
                Y_0
                +
                \int_0^t
                K(t-s)\bigl(-\kappa(Y_s-\theta)\bigr)\,ds
                +
                \int_0^t
                K(t-s)\nu\sqrt{Y_s}\,dB_s,
        \end{align}
        where $K$ is a locally square-integrable kernel. For the constant kernel $K(t)=1$, we recover the classical Heston model above, choosing $K(t)=\sum_{i=1}^N c_ie^{-\lambda_i t}$
yields the lifted Heston model of \textcite{AbiJaber2019a},  while the fractional kernel
\begin{align}\label{eq:kernelfrac}
K(t)=\frac{t^{h-\frac12}}{\Gamma\left(h+\frac12\right)}, 
\end{align} with $h\in(0,\tfrac12)$
gives the rough Heston model of \textcite{ElEuch2019}. In general, the factor process $Y$ need not be Markovian and, for sufficiently singular kernels such as the fractional kernel, may even fail to be a semimartingale. For such kernels, the existence of a strong solution remains an open problem. In this case, one may instead assume that $Y$ and $B$ are adapted to a common filtration $\mathbb{G}$, while $B^\perp$ remains independent of $\mathbb{G}$. For continuously differentiable kernels $K$, strong existence and uniqueness holds for $Y$, see \textcite[Proposition B.3]{abi2019multifactor}.
        \end{itemize}
\end{example}

        \begin{example}[Singular variance clock models]
                \label{setting:singular_clocks}
                The IVC framework also includes singular variance clocks. 
                A notable example arises from Volterra Heston models whose kernel $K$ belongs to $L^1_{\mathrm{loc}}$ but not $L^2_{\mathrm{loc}}$. 
                In this case, \eqref{eq:exvolterra} fails to be well-defined.
                Still, the equation \[
                        U_t = \int_{0}^{t} \left(Y_0 + \int_{0}^{s} K(s-u) \kappa \theta \d u\right) \d s + \int_0^t K(t-s) \left(-\kappa U_s + \nu M_s \right) \d s
                \] that arises from \eqref{eq:exvolterra} by the stochastic Fubini theorem remains well-posed (see \textcite{AbiJaber2021}).
                The (weak) solution to this equation is a tuple $(U, M)$, where $M$ is a continuous $\mathbb{F}^{M}$-local martingale with quadratic variation $U$.
                For the orthogonal term, let $W^\perp$ be a Brownian motion independent of $\mathbb{F}^{M}$ and set $M^\perp_t = W^\perp_{U_t}$.
                Setting $\mathbb{F} = \mathbb{F}^{M} \vee \mathbb{F}^{M^\perp}$, $M$ and $M^\perp$ are strongly orthogonal $\mathbb{F}$-local martingales and the risky asset price process is defined by \eqref{eq:IVC}, i.e.\ \[
                        \frac{dS_t}{S_t} = r dt + \Lambda dU_t + \rho dM_T + \sqrt{1-\rho^2} dM^\perp_t
                \] for some $\mathbb{F}^M$-progressively measurable processes $r, \Lambda$.
                When $K$ lies in $L^2_{\mathrm{loc}}$, $U$ is absolutely continuous w.r.t.\ the Lebesgue measure and can be written as $U_t = \int_0^t Y_s \d s$ with $Y$ given by \eqref{eq:exvolterra}.
                When $K$ lies in $L^1_{\mathrm{loc}}$ but not $L^2_{\mathrm{loc}}$, the integrated variance process $U$ is singular with respect to the Lebesgue measure and no spot variance process exists.
                In particular, for the fractional kernel \eqref{eq:kernelfrac}, the model extends naturally to negative Hurst indices $h\in(-1/2,0]$, thereby recovering the hyper-rough Heston model introduced by \textcite{Jusselin2020}. 
                Such models are not covered by the stochastic volatility framework of \cref{setting:typical_factor_model}, but do fall in the general IVC framework provided in \eqref{eq:IVC} as we illustrate in \cref{section:heston}.
                More generally, one may consider singular Volterra clock models beyond the affine Heston setting as recently developed in \textcite{abijaber2026volterra}.
                In this class of models, the integrated variance process is given by $U = f(A)$, where $A$ is defined by the equation \[
                        A_t = G_0(t) + \int_0^t K(t-s) \left(-\kappa A_s + M_s \right) \d s
                \] and $M$ is a local martingale with $\langle M \rangle_t = f(A_t) $.
                Here, $f$ and $G_0$ are suitable functions.
                For the affine specification $f(a) = \nu a$, this reduces to the Volterra Heston model, but can model more general path-dependent diffusivity structures for different choices of $f$.
        \end{example}
        
		The agent allocates a fraction of wealth $\Pi = (\Pi_t)_{t\ge 0}$ into the risky asset $S$, and consumes at a rate equal to the fraction of wealth $\Xi = (\Xi_t)_{t\ge 0}$.
		The wealth process $X^{\Pi, \Xi}$ has dynamics 
        \begin{equation}
            \label{setting:wealth_dynamics}
            \frac{dX^{\Pi, \Xi}_t}{X^{\Pi, \Xi}_t} = (r_t - \Xi_t) dt + \Pi_t \Lambda_t dU_t + \Pi_t \rho dM_t + \Pi_t \sqrt{1-\rho^2} dM^{\perp}_t
        .\end{equation}

		The optimal investment and consumption problem is to find \begin{equation}
		\label{setting:control_problem}
				V = \sup_{(\Pi, \Xi) \in \mathscr{A}} \E\left[\int_{0}^{\infty} \exp\left(- \int_{0}^{t} \delta_u \d u\right) \frac{(\Xi_t X^{\Pi, \Xi}_t)^{1-R}}{1-R} \d t \right] 
		,\end{equation} where $\delta$ is the impatience rate (which we assume to be $\mathbb{F}$-progressively measurable with $\int_0^t |\delta_s| \d s < \infty$ a.s.\ for all $t > 0$), $R \in (0, \infty) \setminus \{ 1 \} $ is the relative risk aversion and $\mathscr{A}$ the set of admissible strategies.
		We call a strategy $(\Pi, \Xi)$ admissible if $\Pi$ and $\Xi$ are $\mathbb{F}$-progressively measurable, $\Xi_t \ge 0$ a.s.\ for all $t\ge 0$, and \[
				\int_{0}^{t} \Xi_s \d s + \int_{0}^{t} \Pi_s^2 \d U_s < \infty 
		\] a.s.\ for all $t \ge 0$.\footnote{Note that non-negativity of $X$ follows automatically.}
		We call problem \eqref{setting:control_problem} well-posed if $|V| < \infty$ for all $X_0 > 0$; otherwise the problem is called ill-posed.

        \begin{definition}
                \label{setting:def:strongly_wellposed}
                We call the process \begin{equation}\label{eq:frozenH}
		                  H_t = \int_{0}^{t} \left(\frac{1}{R} \delta_s + \frac{R-1}{R} r_s\right) \d s + \frac{R-1}{2R^2} \int_{0}^{t} \Lambda_s^2 \d U_s 
		          \end{equation}
                the \emph{cumulative myopic consumption}. 
                We call the model \emph{strongly myopically well-posed} if there exists a constant $C > 0$ s.t.\ $H_s - H_t \ge C(s-t)$.
        \end{definition}

		The names are motived by the fact that in the absolutely continuous setting of \cref{setting:typical_factor_model}, we have \[
				H_t = \int_0^{t} \eta(Y_s) \d s 
		,\] where \begin{align}\label{eq:frozencons}
				\eta(y) = \frac{1}{R} \left(\delta(y) - (1-R) \left(r(y) + \frac{\lambda(y)^2}{2R} \right)\right)
		\end{align} is the \emph{myopic consumption rate} and the model is called strongly myopically well-posed if there exists a constant $C > 0$ s.t.\ $\eta \ge C$, see \cite{Gutekunst2025}. 
        In turn, the terminology \enquote{myopic consumption rate} is motivated by the fact that if the stochastic factor $Y$ is frozen at the value $y$ (which turns the market into a Black-Scholes model), $\eta(y)$ is the optimal consumption rate in the frozen market (provided that $\eta(y) > 0$).
        
		We denote by $R_{\rho} = (1-\rho^2) R + \rho^2$ the correlation-adjusted risk aversion.
        Notice that $R > 1$ if and only if $R_{\rho} > 1$ (unless $\rho^2 = 1$, in which case $R_{\rho} = 1$).
        $R_{\rho}$ is a convex combination of $R$ and $1$, getting pulled towards $1$ more strongly the closer the market is to being complete.
		
		We make the following three standing assumptions on the market model.
		Firstly, we assume that the minimal distortion measure is well-defined. 
		\begin{assumption}
				\label{setting:assumption:Q}
				The process $\mathcal{E}\left(\frac{1-R}{R} \rho \Lambda \cdot M\right)$ is an $\mathbb{F}$-martingale.	
		\end{assumption}
		
		By the extension property of $\mathbb{F}$, there exists a probability measure $\Q$ with $\frac{\text{d}\Q}{\text{d}\mathbb{P}} \Bigr|_{\mathcal{F}_t} = \mathcal{E}\left( \frac{1-R}{R} \rho \Lambda \cdot M \right)_t $.

		Secondly, we assume that the process 
        \begin{align}\label{eq:tildeM}
             \tilde{M} = M + \int_0^{\cdot} \frac{R-1}{R} \rho \Lambda_s \d U_s
        \end{align}
        has the predictable representation property over $(\mathbb{F}^{M}, \Q)$.
        Note that $\tilde{M}$ is a $\Q$-local martingale by Girsanov's theorem.

		\begin{assumption}
				\label{setting:assumption:M-PRP}
				For every $(\mathbb{F}^{M}, \Q)$-local martingale $N$ there exists an $\mathbb{F}^{M}$-predictable process $Z$ s.t.\ \[
						N_t = N_0 + \int_{0}^{t} Z_s \d{} \tilde{M}_s
				.\] 
		\end{assumption}

		In particular, this means that every $(\mathbb{F}^{M}, \Q)$-local martingale is also an $(\mathbb{F}, \Q)$-local martingale, i.e.\ $\mathbb{F}^{M}$ is immersed in $\mathbb{F}$ (under $\Q$).
		This means that we have $ \E^{\Q}\left[ Y \condbar \mathcal{F}_t \right] = \E^{\Q}\left[ Y \condbar \mathcal{F}^{M}_t \right] $ for all $Y \in L^{1}(\mathcal{F}^{M}_{\infty})$ (see \cite[Thm.~3.2]{Aksamit2017}).

		Finally, we assume that the cumulative myopic consumption is determined by $M$.
        This enforces the usual correlation structure of the model coefficients being determined by $M$, while $M^\perp$ is uncorrelated noise.
		\begin{assumption}
				\label{setting:assumption:M-adapted}
				The processes $H$ and $\Lambda$ are $\mathbb{F}^{M}$-adapted.
		\end{assumption}

		\begin{remark}
				\Cref{setting:assumption:M-PRP,setting:assumption:M-adapted} are not needed if $\mathbb{F}^M$ can be embedded in the Brownian filtration, see \cref{verification:brownian_filtration}.
				We work in a more general non-Brownian framework to accommodate the Volterra Heston model, where the existence of a strong solution for singular kernels is an open problem, as well as models with singular integrated variance clocks (since in the Brownian filtration the quadratic variation of a local martingale is always absolutely continuous).
		\end{remark}

		\section{General BSDE theory associated with the IVC framework}
		\label{section:hjb}

        In this section, we introduce and study the non-standard backward stochastic differential equation associated with the investment and consumption problem in the general IVC framework \eqref{eq:IVC}. We first motivate its form and establish general well-posedness and stability results. In Section~\ref{section:verification}, we use the BSDE to derive the verification theorem for the optimal control problem.

		\subsection{Motivation and Preliminaries of Infinite-Horizon IVC-BSDEs}
		\label{subsection:hjb:derivation}

        We first motivate the backward equation studied in this section by
considering the classical  Markovian stochastic volatility setting of
\cref{setting:typical_factor_model}. Suppose that the factor process
$Y$ in \eqref{setting:typical_factor_model_dynamics}  is an Itô diffusion of the form 
$$ dY_t = a(Y_t) dt + b(Y_t) dB_t,$$
for given functions $a$ and $b$. In this case, the value function $(x,y)\mapsto V(x,y)$ associated with the control problem \eqref{setting:control_problem} is
formally expected to solve the Hamilton--Jacobi--Bellman (HJB) equation
	\begin{multline}
				\label{hjb:classical}
				\sup_{(\pi, \xi) \in \R \times (0, \infty)} \bigg\{ \frac{(\xi x)^{1-R}}{1-R} + (r + \pi \lambda \sigma - \xi) x \frac{\partial V}{\partial x} + \frac{1}{2} \pi^2 \sigma^2 x^2 \frac{\partial ^2 V}{\partial x^2} \\ + a \frac{\partial V}{\partial y} + \frac{1}{2} b^2 \frac{\partial ^2 V}{\partial y^2} + \rho b \pi \sigma x \frac{\partial ^2 V}{\partial x \partial y} - \delta V \bigg\} = 0
		.\end{multline} 
		After solving the first-order conditions and using the distortion transformation $V(x, y) = \frac{x^{1-R}}{1-R} f(y)^{\frac{R}{R_{\rho}}}$, where $R_{\rho} = (1-\rho^2) R + \rho^2$ is the correlation-adjusted risk aversion, the HJB equation becomes \[
				\frac{1}{2} b^2 f'' + \left(a + \frac{1-R}{R} \rho \lambda b \right) f' - R_{\rho} \eta f + R_{\rho} f^{1-\frac{1}{R_{\rho}}} = 0
		.\] 
		By Feynman-Kac, $f$ has the (implicit) stochastic representation \begin{align}\label{eq:storepf}
				f(Y_t) = \E^{\Q}\left[ \int_{t}^{\infty} \exp\left(- \int_{t}^{s} R_{\rho} \eta(Y_u) \d u \right) R_{\rho} f(Y_s)^{1-\frac{1}{R_{\rho}}} \condbar Y_t \right], 
		\end{align}
        with $\eta$ the myopic consumption rate defined in \eqref{eq:frozencons}.
		
        In the general IVC setting \eqref{eq:IVC}, an underlying factor process $Y$ need not
exist. Even when such a process exists, it is generally non-Markovian, as is the case in stochastic Volterra models of the form \eqref{eq:exvolterra}. Consequently, the value function cannot, in general, be characterized by a finite-dimensional HJB equation. This motivates a direct backward formulation in terms of the integrated variance clock. More precisely, the preceding stochastic representation \eqref{eq:storepf} suggests the following infinite-horizon implicit stochastic equation:
\begin{equation}
				\label{hjb}
				F_t = \E^{\Q}\left[\int_{t}^{\infty} \exp\left( -R_{\rho} (H_s - H_t) \right) R_{\rho} F_s^{1-\frac{1}{R_{\rho}}} \condbar \mathcal{F}_t \right], 
		\end{equation}
        with $H$ the cumulative myopic consumption defined in \eqref{eq:frozenH}. 	We note that in the complete market setting, i.e.~$\rho=\pm 1$, we have that  $R_{\rho} = 1$, so that \eqref{hjb} explicitly defines $F$. More generally, the stochastic representation \eqref{hjb} corresponds formally to the infinite-horizon BSDE
\begin{align}\label{eq:IVCBSDE}
	dF_t = -R_{\rho} F_t^{1-\frac{1}{R_{\rho}}} \, dt + R_{\rho} F_t \, dH_t + Z^F_t \, d\tilde{M}_t,
\end{align}
where $\tilde{M}$ is the $\Q$-local martingale defined in \eqref{eq:tildeM}. Observe that both the finite-variation and martingale components are naturally expressed in terms of both calendar time and the integrated variance clock: the process $H$ contains a finite-variation term driven by $U$, while $\tilde{M}$ is driven by the time-changed Brownian motion $W_U$. This means that we recover non-standard BSDEs of the form \eqref{eq:introIVCBSDE} where the finite-variation measure $dU$ may be singular with respect to Lebesgue measure. It is not surprising that $dt$ and $dU$ terms both enter in the BSDE: The market uncertainty is driven by $U$, whereas consumption happens w.r.t.\ the calendar clock. We refer to equations of the form \eqref{hjb} or \eqref{eq:IVCBSDE} as \emph{integrated variance clock BSDEs} (IVC-BSDEs). The equivalence between the stochastic representation \eqref{hjb} and the differential formulation \eqref{eq:IVCBSDE} is established in Lemma~\ref{existence:dynamics} below.

The relevance of the IVC-BSDE is that its solution completely characterizes the investment--consumption problem \eqref{setting:control_problem} in the IVC framework. Indeed, if $(F,Z^F)$ solves the IVC-BSDE \eqref{eq:IVCBSDE}, then the candidate value process is given by
\begin{align}
	V_t = \frac{X_t^{1-R}}{1-R} F_t^{\frac{R}{R_{\rho}}},
\end{align}
while the candidate optimal consumption and investment strategies are
\begin{align}
	\hat{\Xi}_t = F_t^{-\frac{1}{R_{\rho}}}, \qquad
	\hat{\Pi}_t = \frac{1}{R}\Lambda_t + \frac{\rho}{R_{\rho}} \frac{Z^F_t}{F_t}.
\end{align}
Thus, solving the IVC-BSDE is the key step in solving the investment--consumption problem. The rigorous justification of these expressions is deferred to the verification argument in Theorem~\ref{verification:verification}.

We now develop the theory of the IVC-BSDE \eqref{eq:IVCBSDE}. For this, we denote by $\mathscr{P}_{> 0}$ the set of positive $\mathbb{F}$-progressively measurable processes, and we
		define the (non-linear) operator $T: \mathcal{D}(T) \to \mathscr{P}_{>0} $ by\footnote{The conditional expectation here is to be understood as the optional projection.} \begin{align}\label{eq:defT}
                F \mapsto TF, \quad (TF)_t = \E^{\Q}\left[ \int_{t}^{\infty} \exp\left( -R_{\rho} (H_s - H_t) \right)  R_{\rho} F_s^{1-\frac{1}{R_{\rho}}} \d s \condbar \mathcal{F}_t \right]
        ,\end{align} where $\mathcal{D}(T)$ is the convex cone \[
        		\mathcal{D}(T) = \left\{F \in \mathscr{P}_{>0}: \E^{\Q}\left[ \int_{t}^{\infty} \exp\left( -R_{\rho} (H_s - H_t) \right) R_{\rho} F_s^{1-\frac{1}{R_{\rho}}} \d s \right] < \infty \text{ for all } t \ge 0 \right\} 
        .\] 
		Note that the positive solutions to \eqref{hjb} are precisely the fixed points of $T$.

We begin by establishing the equivalence between the integrated formulation \eqref{hjb} and its differential counterpart \eqref{eq:IVCBSDE}. Specifically, a process is a solution to the integrated equation if and only if it satisfies the differential formulation together with appropriate martingality and transversality conditions.
        The proof of \cref{existence:dynamics} can be found in \cref{appendix:auxiliary}.

        \begin{lemma}
            \label{existence:dynamics}
            Set $D_t = \exp(-R_{\rho} H_t ) $.
            An $\mathbb{F}^{M}$-adapted process $F \in \mathcal{D}(T)$ is a solution to \eqref{hjb} if and only if there exists an $\mathbb{F}^{M}$-predictable process $Z^F$ s.t.\ 
            \begin{enumerate}[(i)]
                    \item \label{existence:dynamics:drift} Equation \eqref{eq:IVCBSDE} holds for $(F,Z^F)$,
                    \item \label{existence:dynamics:martingale} $\int_0^{\cdot } D_s Z^F_s \d{} \tilde{M}_s$ is an $(\mathbb{F}, \Q)$-martingale, and
                    \item \label{existence:dynamics:transversality} $\liminf_{T\to \infty} \E^{\Q}[D_T F_T | \mathcal{F}_t ] = 0$ for all $t \ge 0$.
            \end{enumerate}
		\end{lemma}

        In the following section on existence and uniqueness, it will be useful to also consider sub- and supersolutions to \eqref{hjb}.
        A process $F$ is called a \emph{subsolution} if \[  
                F_t \le \E^{\Q}\left[\int_{t}^{\infty} \exp\left( -R_{\rho} (H_s - H_t) \right) R_{\rho} F_s^{1-\frac{1}{R_{\rho}}} \condbar \mathcal{F}_t \right] 
        .\]
        Similarly, $F$ is called a \emph{supersolution} if \[
                F_t \ge \E^{\Q}\left[\int_{t}^{\infty} \exp\left( -R_{\rho} (H_s - H_t) \right) R_{\rho} F_s^{1-\frac{1}{R_{\rho}}} \condbar \mathcal{F}_t \right] 
        .\]
        Sub- and supersolutions can also be characterised in differential form. 
        We say that a process $A$ strongly majorises another process $B$ if $A-B$ is increasing.
        We omit the proof of \cref{existence:dynamics:sub_super_solution} as it is analogous to the proof of \cref{existence:dynamics}.
        
        \begin{corollary}
            \label{existence:dynamics:sub_super_solution}
            Set $D_t = \exp\left(-R_{\rho} H_t \right) $.
            An $\mathbb{F}^{M}$-adapted process $F \in \mathcal{D}(T)$ with dynamics \[
                    dF_t = dA_t + Z^F_t d\tilde{M}_t 
            ,\] where $A$ is a finite-variation process and $Z^F$ is $\mathbb{F}^{M}$-predictable, is a subsolution (supersolution) to \eqref{hjb} if  
            \begin{enumerate}[(i)]
                    \item \label{existence:dynamics:sub_super_solution:drift} $\int_{0}^{\cdot } -R_{\rho} F_s^{1-\frac{1}{R_{\rho}}} \d s + \int_{0}^{\cdot } R_{\rho} F_s \d H_s$ is strongly majorised by $A$ (strongly majorises $A$),
                    \item \label{existence:dynamics:sub_super_solution:martingale} $\int_0^{\cdot } D_s Z^F_s \d{} \tilde{M}_s$ is an $(\mathbb{F}, \Q)$-submartingale (supermartingale), and
                    \item \label{existence:dynamics:sub_super_solution:transversality} $\liminf_{T\to \infty} \E^{\Q}[D_T F_T | \mathcal{F}_t ] \le 0$ ($\ge 0$) for all $t \ge 0$.
            \end{enumerate}
		\end{corollary}
        
        \begin{remark}
		      \label{existence:dynamics:sub_super_solution:sufficient_criteria}
			For a supersolution, condition \ref{existence:dynamics:sub_super_solution:martingale} is always satisfied since $\int_{0}^{t} D_s Z^F_s \d{} \tilde{M}_s \ge N_t - N_0 \ge -N_0 $ is bounded from below (where $N$ is a non-negative process from the proof of \cref{existence:dynamics}).
			Moreover, condition \ref{existence:dynamics:sub_super_solution:transversality} is trivially satisfied since $D_T, F_T \ge 0$ for all $T\ge 0$.
			For a subsolution, condition \ref{existence:dynamics:sub_super_solution:transversality} is satisfied in particular if the model is strongly myopically well-posed and $F$ is bounded.
		\end{remark}
        
		\subsection{Existence and Uniqueness}
		\label{subsection:hjb:existence}

        In this section, we study the existence and uniqueness of solutions to the infinite-horizon IVC-BSDE \eqref{hjb} (equivalently \eqref{eq:IVCBSDE}). The first main result establishes existence and uniqueness for general strongly myopically well-posed models.
        
        \begin{theorem}
            \label{existence:uniformly_wellposed}

            Assume that $R > 1$ and the model is strongly myopically well-posed in the sense of \cref{setting:def:strongly_wellposed}.
            Set \[
                    G^{1}_t = \E^{\Q}\left[ \int_{t}^{\infty} \exp(-(H_s - H_t)) \d s \condbar \mathcal{F}_t \right]^{R_{\rho}}
            .\] 
            Then \eqref{hjb} has a unique bounded solution $F$ that satisfies $G^{1} \le F$.
            Moreover, $F$ is $\mathbb{F}^{M}$-adapted.
        \end{theorem}

	    Moving beyond the strongly myopically well-posed setting, our second main result gives a necessary and sufficient criterion for the existence of a solution that is bounded and bounded away from zero if the growth of the cumulative myopic consumption is bounded from above (e.g.\ if $\eta$ is bounded from above in the absolutely continuous setting of \cref{setting:typical_factor_model}).
        We denote by $\one$ by constant process with value $1$.

        \begin{theorem}
                \label{existence:eta_bounded}
               Assume that $R > 1$ and that there exists a constant $C > 0$ s.t.\ $H_s - H_t \le C (s-t)$ for all $0 \le t \le s$.
                Then there exists a solution $F \in \mathcal{D}(T)$ to \eqref{hjb} that is bounded and bounded away from zero if and only if $\one \in \mathcal{D}(T)$ and $T \one$ is bounded.
                Moreover, if it exists, $F$ is the unique such solution and is $\mathbb{F}^{M}$-adapted.
		\end{theorem}

        Proofs of \cref{existence:uniformly_wellposed,existence:eta_bounded} are postponed until the end of the section.
        
        \begin{remark}
                While Theorem~\ref{existence:eta_bounded} does not explicitly require the growth of the cumulative myopic consumption to also be bounded from below, we heuristically expect a bounded solution to exist only in that case. 
                For example, see the results on the Vasicek model in \textcite{Guasoni2019,Gutekunst2025}, where the myopic consumption rate becomes arbitrarily negative when the interest rate is negative and the optimal consumption rate consequently approaches $0$, which corresponds to an unbounded solution to \eqref{hjb}.
        \end{remark}

        \begin{remark}
                In a Markovian setting, the existence statement of \cref{existence:uniformly_wellposed} was shown by \textcite[Thm.~4.3]{Gutekunst2025}.
                \Cref{existence:eta_bounded} generalises the necessary and sufficient existence criterion obtained in \cite[Thms.~3.4,A.1]{Gutekunst2025} for the case where the stochastic factor is a continuous-time Markov chain with finite state space or a reflected diffusion.
                The main interest of \cref{existence:eta_bounded} is that it provides a relatively simple necessary condition for the existence of a solution.
                In this setting, we will also be able to verify that the candidate solution is indeed optimal, see \cref{verification:example:bounded_eta}.
                
                Even in the Markovian case of \cref{setting:typical_factor_model}, \cref{existence:eta_bounded,verification:example:bounded_eta} yield a more general existence, uniqueness, and verification statement than the current literature.
                In the treatment of the bounded coefficient case in \cite[Sec.~7.1]{Gutekunst2025}, the drift and diffusion coefficients of the stochastic factor are assumed to be bounded; moreover the diffusion coefficient is assumed to be bounded away from zero. 
                Here, this type of assumption is not necessary; we only make assumptions on the myopic consumption rate.
        \end{remark}

        \begin{remark}
                While the models that are directly of interest typically have an unbounded myopic consumption rate (for instance, this is the case in the Heston, Kim-Omberg, and Vasicek models), a numerical computation of the solution will always have to introduce some kind of truncation, which then brings the model into the setting of \cref{existence:eta_bounded,verification:example:bounded_eta}.
                For example, the numerical approach of \cite{Gutekunst2025} implements truncation by reflecting the stochastic factor at the boundary of some large, but finite, interval. This yields a model with bounded myopic consumption rate as long as the myopic consumption rate in the original model is locally bounded.
                Hence, understanding the case of a bounded myopic consumption rate is of interest for the study of numerical methods.
        \end{remark}
        
        The rest of this section is devoted to the proofs of Theorems~\ref{existence:uniformly_wellposed} and \ref{existence:eta_bounded}.
    
		We begin by collecting some basic properties of the operator $T$ defined in \eqref{eq:defT}.

		\begin{lemma}
				\label{existence:operator:properties}
				Let $F \in \mathcal{D}(T)$.
				\begin{enumerate}[(i)]
						\item \label{existence:operator:properties:homogeneous}
								Homogeneity: For all constants $C > 0$, we have $T(CF) = C^{1-\frac{1}{R_{\rho}}} TF$.
						\item \label{existence:operator:properties:adapted}
								If $F$ is $\mathbb{F}^{M}$-adapted, then $TF$ is also $\mathbb{F}^{M}$-adapted and has continuous sample paths.
						\item \label{existence:operator:propeties:monotone}
								Assume that $R > 1$. If $G \in \mathscr{P}_{>0}$ with $G \le F$, then $G \in \mathcal{D}(T)$ and $TG \le TF$.
						\item \label{existence:operator:properties:hoelder}
								Hölder inequality: If $F^{p}, G^{q} \in \mathcal{D}(T)$ for some $G \in \mathscr{P}_{>0}$, where $p, q > 1$ with $\frac{1}{p} + \frac{1}{q} = 1$, then $FG \in \mathcal{D}(T)$ and $T(FG) \le (TF^{p})^{\frac{1}{p}} (TG^{q})^{\frac{1}{q}}$.
				\end{enumerate}
		\end{lemma}
		\begin{proof}
				\ref{existence:operator:properties:homogeneous} follows immediately from the definition of $T$. 

				For \ref{existence:operator:properties:adapted}, note that $ \E^{\Q}\left[ Y \condbar \mathcal{F}_t \right] = \E^{\Q}\left[ Y \condbar \mathcal{F}^{M}_t \right] $ for all $Y \in L^{1}(\mathcal{F}^{M}_\infty)$ by \cref{setting:assumption:M-PRP}.
				If $F$ is $\mathbb{F}^{M}$-adapted, we hence have \[
						(TF)_t = \E^{\Q}\left[ \int_{t}^{\infty} \exp(-R_{\rho} (H_s - H_t)) R_{\rho} F_s^{1-\frac{1}{R_{\rho}}} \d s \condbar \mathcal{F}^{M}_t \right] 
				,\] so $TF$ is $\mathbb{F}^{M}$-adapted.
				Moreover, $TF$ has continuous sample paths since $H$ has continuous sample paths and $\mathbb{F}^{M}$ is continuous by \cref{setting:assumption:M-PRP}.
                
				\ref{existence:operator:propeties:monotone} follows since the map $x \mapsto x^{1-\frac{1}{R_{\rho}}}$ is non-decreasing when $R_{\rho} \ge 1$, which is always the case when $R > 1$.

				Finally, \ref{existence:operator:properties:hoelder} follows by first applying the Hölder inequality pathwise to the inner integral (w.r.t.\ the measure $\exp(-R_{\rho} (H_s - H_t)) R_{\rho} \d s$) and then applying the Hölder inequality for conditional expectations.
		\end{proof}

		Our existence theory is based on the method of sub- and supersolutions.
		Recall that a process $F \in \mathcal{D}(T)$ is called a \emph{subsolution} to \eqref{hjb} if $F \le TF$.
		Similarly, $F$ is called a \emph{supersolution} to \eqref{hjb} if $F \ge TF$.
		Given an ordered pair of sub- and supersolutions, we can construct a solution that lies between the sub- and supersolution. 
		In fact, we can even construct maximal and minimal solutions.
        
		Since the construction relies on the monotonicity of $T$, we mostly restrict to the case $R > 1$ in this section.
        This is somewhat atypical, usually the case $R > 1$ is regarded as being more difficult than $R \in (0, 1)$. 
        This is indeed the case for primal verification arguments, but the opposite seems to be the case for a direct study of the HJB equation in the stochastic factor setting.
        This behaviour can also be observed in \cite[Sec.~4.1,Prop.~5.1]{Gutekunst2025}, where certain results do not apply for either $R \in (0, 1)$ or $R_\rho \in (0, \frac{1}{2}]$.
        
        \begin{theorem}
				\label{existence:fixed_point_iteration}
                Assume that $R>1$.
				Let $G^1, G^2 \in \mathcal{D}(T)$ with $G^1 \le G^2$ be sub- and supersolutions to \eqref{hjb}, respectively.
				Define the sequences $\ubar{F}^{1} = G^{1}$, $\ubar{F}^{n+1} = T\ubar{F}^{n}$, and $\bar{F}^1 = G^2$, $\bar{F}^{n+1} = T\bar{F}^{n}$, where $n \in \N$.
				Then $\ubar{F}^{n} \uparrow \ubar{F}$ and $\bar{F}^{n} \downarrow \bar{F}$ as $n \to \infty$, where $\ubar{F}$ ($\bar{F}$) is the minimal (maximal) solution to \eqref{hjb} with $G^{1} \le \cdot \le G^2$.
				In particular, \eqref{hjb} has a unique solution that satisfies $G^{1} \le \cdot \le G^2$ if and only if $\ubar{F} = \bar{F}$.
				Moreover, if $G^{1}$ ($G^2$) is $\mathbb{F}^{M}$-adapted, then $\ubar{F}$ ($\bar{F}$) is also $\mathbb{F}^{M}$-adapted.
        \end{theorem}
        \begin{proof}
				We consider only $\bar{F}$, the properties of $\ubar{F}$ follow analogously.
				As $G^2$ is a supersolution, we have $\bar{F}^1 \ge T\bar{F}^1 = \bar{F}^2$.
				In particular, we also have $\bar{F}^2 \in \mathcal{D}(T)$ by \cref{existence:operator:properties}.
				Inductively, the monotonicity of $T$ (see \cref{existence:operator:properties}) yields that $\bar{F}^{n} \ge \bar{F}^{n+1}$ and $\bar{F}^{n+1} \in \mathcal{D}(T)$ for all $n \in \N$, i.e.\ the sequence $(\bar{F}^n)_{n\in\N}$ is decreasing.
                Since $\bar{F}^1 \ge G^{1}$ and $G^{1}$ is a subsolution, it moreover follows inductively that $\bar{F}^n \ge G^{1}$ for all $n \in \N$.
                Hence, there exists a process $\bar{F}$ with $\bar{F} \ge G^1$ s.t.\ $\bar{F}^n \downarrow \bar{F}$.
				Note that $\bar{F} \in \mathcal{D}(T)$ as $\bar{F} \le \bar{F}^1$ by \cref{existence:operator:properties}.
                By dominated convergence, it follows that $T\bar{F}^n \downarrow T\bar{F}$.
				Hence, we have $\bar{F} = \lim_{n \to \infty} \bar{F}^n = \lim_{n \to \infty} T\bar{F}^n = T\bar{F}$, so $\bar{F}$ is a solution to \eqref{hjb}.

				For any solution $F$ to \eqref{hjb} with $F \le G^2$, \cref{existence:operator:properties} yields $F = TF \le TG^2 = \bar{F}^2$.
				Inductively, we have $F \le \bar{F}^{n}$ for all $n \in \N$, and thus also $F \le \bar{F}$.
				Hence, $\bar{F}$ is maximal.

				If $G^2$ is $\mathbb{F}^{M}$-adapted, $\bar{F}^{n}$ is $\mathbb{F}^{M}$-adapted for all $n \in \N$ by \cref{existence:operator:properties}.
				Hence, $\bar{F}$ is also $\mathbb{F}^{M}$-adapted.
        \end{proof}

		Now, we turn to finding sub- and supersolutions. 
		After rescaling, the solution from the complete case $R_{\rho} = 1$ is a subsolution to the incomplete case $R_{\rho} > 1$. 
		In the absolutely continuous setting of \cref{setting:typical_factor_model}, the same result can be obtained using subsolution methods for ODEs when $Y$ is an Itô diffusion, see \cite[Cor.~4.2]{Gutekunst2025}.

		\begin{proposition}
				\label{existence:subsolution}
				Assume that $R>1$ and set \[
						G_t = \E^{\Q}\left[ \int_{t}^{\infty} \exp(-(H_s - H_t)) \d s \condbar \mathcal{F}_t \right]^{R_{\rho}} 
				.\] 
				Denote by $Z^G$ the control process of $G$, and assume that $G \in \mathcal{D}(T)$, $\int_{0}^{\cdot} \exp(-R_{\rho} H_s) Z^G_s \d{} \tilde{M}_s$ is a $(\mathbb{F}, \Q)$-submartingale, and that $\liminf_{T \to \infty} \E^{\Q}\left[ \exp(-R_{\rho} H_T) G_T \condbar \mathcal{F}_t \right] = 0 $.
				Then $G$ is an $\mathbb{F}^{M}$-adapted subsolution.
		\end{proposition}
		\begin{proof}
				Set \[
						J_t = G_t^{\frac{1}{R_{\rho}}} = \E^{\Q}\left[ \int_{t}^{\infty} \exp(-(H_s - H_t)) \d s \condbar \mathcal{F}_t \right]
				.\]
				Notice that $J$ is $\mathbb{F}^{M}$-adapted by \cref{existence:operator:properties}.
				Moreover, notice that $J$ is the solution to \eqref{hjb} for $R_{\rho} = 1$, i.e.\ in the complete market case.
				Using \cref{existence:dynamics} for $J$ with $R_{\rho} = 1$, the dynamics of $J$ are given by \[
						dJ_t = -dt + J_t dH_t + Z^{J}_t d\tilde{M}_t
				\] for some $\mathbb{F}^{M}$-predictable process $Z^{J}$.
				Hence, Itô's lemma yields that the dynamics of $G$ are given by \[
						dG_t = -R_{\rho} G_t^{1-\frac{1}{R_{\rho}}} dt + R_{\rho} G_t dH_t + \frac{1}{2} R_{\rho} (R_{\rho}-1) G_t^{1-\frac{2}{R_{\rho}}} (Z^{J}_t)^2 d U_t + R_{\rho} G_t^{1-\frac{1}{R_{\rho}}} Z^{J}_t d\tilde{M}_t
				.\] 
				Since $U $ is non-decreasing and $R_{\rho} \ge 1$, condition \ref{existence:dynamics:sub_super_solution:drift} of \cref{existence:dynamics:sub_super_solution} is satisfied.
				The other conditions of \cref{existence:dynamics:sub_super_solution} are satisfied by assumption.
				Hence, $G$ is a subsolution.
		\end{proof}

		\begin{remark}
				\label{existence:subsolution:sufficient_conditions}
                The conditions of \cref{existence:subsolution} are satisfied if the model is strongly myopically well-posed:
	            In this case, $J$ is bounded, so $\liminf_{T \to \infty} \E^{\Q}\left[ \exp(-R_\rho H_T) J_T^{R_\rho} \condbar \mathcal{F}_t \right] = 0$.
				Moreover, $\int_{0}^{\cdot} \exp(-H_s) Z^{J}_s \d{} \tilde{M}_s $ is a bounded martingale by \cref{existence:dynamics}, so $ \E^{\Q}[\int_{0}^{t} (\exp(-H_s) Z^{J}_s)^2 \d U_s ] < \infty$ for all $t \ge 0$.
				Since $J$ is bounded and $H \ge 0$, we thus have $ \E^{\Q}[\int_{0}^{t} (\exp(-R_\rho H_s) R_\rho J_s^{R_\rho-1} Z^{J}_s)^2 \d U_s ] < \infty $ for all $t \ge 0$.
				Hence, $\int_{0}^{\cdot} \exp(-R_\rho H_s) R_\rho J_s^{R_\rho-1} Z^{J}_s \d{} \tilde{M}_s $ is a true martingale.
		\end{remark}
		
		In analogy to \cite[Cor.~4.2]{Gutekunst2025}, we expect the fully incomplete case $R_{\rho} = R$, i.e.\ $\rho = 0$, to yield a supersolution. 
		However, this supersolution is not particularly useful since (unlike in the complete case) there is no explicit formula available for the solution in the fully incomplete case.
		Instead, we pass to the complete case (which pulls the solution downwards) and decrease the myopic consumption rate (which pulls the solution upwards). 
		As long the myopic consumption rate is decreased sufficiently, this yields a supersolution.
		When $Y$ is an Itô diffusion in the absolutely continuous setting of \cref{setting:typical_factor_model}, this is the lower market contraction of \textcite{Guasoni2020}.

		\begin{proposition}
				\label{existence:supersolution}
				Let $\tilde{H}$ be a process that is strongly majorised by $H$.
				Assume that \[
						G_t = \E^{\Q}\left[\int_{t}^{\infty} \exp\left(- (\tilde{H}_s - \tilde{H}_t) \right) \d s \condbar \mathcal{F}_t \right]^{R_{\rho}}
				\] is well-defined, and that $G \in \mathcal{D}(T)$.
				Denote by $Z^G$ the control process of $G$.
				Assume that $\int_{0}^{\cdot } R_{\rho} G_s \d{} \tilde{H}_s + \int_{0}^{\cdot } \frac{1}{2} \frac{R_{\rho}-1}{R_{\rho}} \frac{(Z^G_s)^2}{G_s} \d{} U_s $ is strongly majorised by $\int_{0}^{\cdot } R_{\rho} G_s \d H_s $.
				Then $G$ is an $\mathbb{F}^{M}$-adapted supersolution to \eqref{hjb}.
		\end{proposition}
		\begin{proof}
				Analogously as in \cref{existence:subsolution}, one sees that the dynamics of $G$ are given by \[
						dG_t = -R_{\rho} G_t^{1-\frac{1}{R_{\rho}}} dt + RG_t d\tilde{H}_t + \frac{1}{2} \frac{R_{\rho}-1}{R_{\rho}} \frac{(Z^G_t)^2}{G_t} d U_t + Z^G_t d\tilde{M}_t
				.\] 
				Since $\int_{0}^{\cdot } R_{\rho} G_s \d{} \tilde{H}_s + \int_{0}^{\cdot } \frac{1}{2} \frac{R_{\rho}-1}{R_{\rho}} \frac{(Z^G_s)^2}{G_s} \d{} U_s $ is strongly majorised by $\int_{0}^{\cdot } R_{\rho} G_s \d H_s $ by assumption, $G$ is a supersolution by \cref{existence:dynamics:sub_super_solution,existence:dynamics:sub_super_solution:sufficient_criteria}.
		\end{proof}

		\begin{remark}
				\label{existence:supersolution:sufficient_conditions}
				In particular, $G$ is well-defined when the $\tilde{H}$-model is strongly myopically well-posed.
		\end{remark}
        
        We are now in a position to prove Theorems~\ref{existence:uniformly_wellposed} and \ref{existence:eta_bounded}.
		
        \begin{proof}[Proof of Theorem~\ref{existence:uniformly_wellposed}]
				Since the model is strongly myopically well-posed, there exists a constant $\eta_0 > 0$ with $H_s - H_t \ge \eta_0 (s-t)$.
                Notice that $T\one \le \frac{1}{\eta_0}$.
				Fix $K \ge \eta_0^{-R_\rho}$, and set $G_2 = K$.
             	By \cref{existence:operator:properties}, we have \[
						TG^2 = K^{1-\frac{1}{R_{\rho}}} T \one \le K^{1 - \frac{1}{R_\rho}} \frac{1}{\eta_0} \le K = G^2
				,\] i.e.\ $G^2$ is a supersolution.
                Moreover, $G^{1}$ is a subsolution by \cref{existence:subsolution,existence:subsolution:sufficient_conditions}.
                Notice that $G^1 \le G^2$.      
				The existence of a $\mathbb{F}^{M}$-adapted solution $F$ with $G^{1} \le F \le K$ now follows from \cref{existence:fixed_point_iteration}.

				It remains to show uniqueness. 
				Let $G^{1} \le \ubar{F} \le \bar{F} \le K$ be the minimal and maximal solutions from \cref{existence:fixed_point_iteration}. 
				Since $F$ is bounded and $H$ is increasing, we have \begin{multline*}
						\bar{F}_t = \E^{\Q}\left[ \int_{t}^{\infty} \exp(-R_{\rho} (H_s - H_t)) R_{\rho} \bar{F}_s^{1-\frac{1}{R_{\rho}}} \d s \condbar \mathcal{F}_t \right] \le R_{\rho} K^{1-\frac{1}{R_{\rho}}} \E^{\Q}\left[ \int_{t}^{\infty} \exp(-(H_s - H_t)) \d s \condbar \mathcal{F}_t \right] \\= R_{\rho} K^{1-\frac{1}{R_{\rho}}} (G^{1}_t)^{\frac{1}{R_{\rho}}} \le R_{\rho} K^{1-\frac{1}{R_{\rho}}} \ubar{F}_t^{\frac{1}{R_{\rho}}}
				.\end{multline*} 
				Moreover, notice that $T^{n} K = K^{ (1-\frac{1}{R_{\rho}})^{n}} T^{n} \one$ by iteratively applying \cref{existence:operator:properties}.
				By \cref{existence:fixed_point_iteration}, we have $T^{n} K \to \bar{F}$ as $n \to \infty$, so that $T^{n} \one \to \bar{F}$ as well since $K^{(1-\frac{1}{R_{\rho}})^{n}} \to 1$.
				Hence, \cref{existence:order} yields $\bar{F} \le \ubar{F}$, so $F$ is the unique solution with $G^{1} \le F \le K$.
				Since $K$ was arbitrary, $F$ is the unique bounded solution with $G^{1} \le F$.
		\end{proof}

		\begin{proof}[Proof of Theorem~\ref{existence:eta_bounded}]
				First, assume that $\one \in \mathcal{D}(T)$ and that $T \one$ is bounded by a constant $K > 0$.
				Since $H_s - H_t \le C (s-t)$ for all $0 \le t \le s$, we have \[
						(T \one)_t = \E^{\Q}\left[ \int_{t}^{\infty} \exp( -R_{\rho} (H_s - H_t) ) R_{\rho} \d s \condbar \mathcal{F}_t\right] \ge \E^{\Q}\left[ \int_{t}^{\infty} \exp(-R_{\rho}C (s-t)) R_{\rho} \d s \condbar \mathcal{F}_t 	\right] = \frac{1}{C}
				.\] 
				Fix $k \le \frac{1}{C}$ and set $G^{1} = k^{R_{\rho}}$ and $G^2 = K^{R_{\rho}}$.
				$G^{1}$ is a subsolution since \[
						T G^{1} = (k^{R_{\rho}})^{1-\frac{1}{R_{\rho}}} T \one \ge k^{R-1} k = k^{R} = G^{1}
				.\]
				Analogously, one sees that $G^2$ is a supersolution.
				Note that we can always choose $k \le K$, so that $G^{1} \le G^2$.
				By \cref{existence:fixed_point_iteration}, there exist a minimal solution $ \ubar{F} $ and a maximal solution $\bar{F}$ to \eqref{hjb} with $G^{1} \le \cdot \le G^2$.
				Notice that $ \ubar{F}, \bar{F}$ are both $\mathbb{F}^{M}$-adapted.
				For uniqueness, it remains to show that $ \ubar{F} = \bar{F}$. 
				It is clear that $ \ubar{F} \le \bar{F}	$ by maximality. 
				Moreover, we have $\bar{F} = \frac{\bar{F}}{\ubar{F}} \ubar{F} \le \frac{K^{R_\rho}}{k^{R_\rho}} \ubar{F}$, so that $\bar{F} \le \ubar{F}$ by \cref{existence:order}.
				Since $k$ ($K$) can be chosen arbitrarily small (large), $F$ is the unique solution that is bounded and bounded away from zero.
				
				Conversely, let $F \in \mathcal{D}(T)$ be a solution to \eqref{hjb} that satisfies $k \le F \le K$ for some constants $k, K > 0$.
				We then have \[
						F_t = \E^{\Q}\left[\int_{t}^{\infty} \exp(-R_{\rho} (H_s - H_t)) R_{\rho} F_s^{1-\frac{1}{R_{\rho}}} \d s \condbar \mathcal{F}_t \right] \ge k^{1-\frac{1}{R_{\rho}}} T \one
				,\] so $\one \in \mathcal{D}(T)$ by \cref{existence:operator:properties}. 
				Moreover, we have $T \one \le \frac{K}{k^{1-\frac{1}{R_{\rho}}}}$, so $T \one$ is bounded.
		\end{proof}

		\subsection{Stability}
		\label{subsection:hjb:stability}

		In this section, we study the stability of solutions to \eqref{hjb} and their associated control process in strongly myopically well-posed models.
        This result will be used in the study of the Volterra Heston model in \cref{section:heston}, where we consider a sequence of regularised approximations to the model.

        In the following, we denote by $C$ the space of continuous functions equipped with the topology of locally uniform convergence.
        
		\begin{theorem}
				\label{stability:uniformly_wellposed}
				Assume that $R>1$.
				Let $(H^{n}, \tilde{M}^{n}, \mathbb{F}^{n})_{n \in \N}$ be a sequence of stochastic factors (which may be defined on different probability spaces).
				Assume that there exists a constant $\eta_0 > 0$ s.t.\ $H^{n}_s - H^{n}_t \ge \eta_0 (s-t)$ for all $n\in \N$, $0 \le t \le s$.
				Moreover, assume that $\sup_n \E^{\Q^{n}}[H^{n}_t]  < \infty $ for all $t \ge 0$.
				Let $F^{n}$ be the unique solution to \eqref{hjb} with control process $Z^{n}$ that satisfies \[
						G^{n}_t := \E^{\Q^{n}}\left[ \int_{t}^{\infty} \exp(-(H^{n}_s - H^{n}_t)) \d s \condbar \mathcal{F}^{n}_t \right]^{R_{\rho}} \le F^{n}_t \le \eta_0^{-R_{\rho}}
				\] for all $t \ge 0$.
				If $(H^{n}, \tilde{M}^{n}) \to (H, \tilde{M}) $ weakly over $C \times C$ w.r.t.\ $\Q^{n}$, then $(F^n, \int Z^n_s \d \langle \tilde{M}^n \rangle_s) \to (F, \int Z_s \d U_s)$ weakly over $C \times C$, where $F$ is the unique solution to \eqref{hjb} with control process $Z$ that satisfies \[
						G_t := \E^{\Q}\left[ \int_{t}^{\infty} \exp(-(H_s - H_t)) \d s \condbar \mathcal{F}_t \right]^{R_{\rho}} \le F_t \le \eta_0^{-R_{\rho}}
				\] for all $t \ge 0$.
		\end{theorem}
		\begin{proof}
				We will show that there exists a subsequence along which the convergence holds. 
				Since the limit is unique and any subsequence again satisfies the assumptions of the theorem, this implies that any subsequence has a convergent subsequence, which yields convergence along the whole sequence.
				
				Set $K = \eta_0^{-R_{\rho}}$. 
				In the following, we denote by $D_K$ the set of càdlàg functions that are bounded by $K$, equipped with the Meyer-Zheng \cite{Meyer1984} topology.
				In addition, we denote $C^{\eta_0} = \{f \in C: f(s) - f(t) \ge \eta_0 (s-t)\} $.

                Set $N^{n}_t = \int_{0}^{t} \exp(-R_{\rho} H^{n}_s) R_{\rho} (F^{n}_s)^{1-\frac{1}{R_{\rho}}} \d s + \exp(-R_{\rho} H^{n}_t) F^{n}_t$.
				By \cref{existence:dynamics}, we have \[
						N^{n} = \E^{\Q^{n}}\left[ \int_{0}^{\infty} \exp(-R_{\rho} H^{n}_s) R_{\rho} (F^{n}_s)^{1-\frac{1}{R_{\rho}}} \d s \condbar \mathcal{F}^{n}_t \right] 
				,\] so $N^{n}$ is a uniformly bounded martingale (independently of $n$) by strong myopic well-posedness and uniform boundedness of $F^{n}$.

				Since $N^{n}$ is a martingale, the conditional variation\footnote{The conditional variation of a process $X$ is defined as $V(X)_t = \sup_{0=t_0<t_1<\dots<t_n=t} \sum_{i=0}^{n-1} \E\left[ \left| \E\left[ X_{t_{i+1}} - X_{t_i} \condbar \mathcal{F}_{t_i} \right] \right|\right] $, see \cite[Eq.~(4)]{Meyer1984}. If $X$ is a martingale, the conditional variation satisfies $V_t(X) = 0$ for all $t$.} of $\exp(-R_{\rho} H^{n}) F^{n}$ satisfies \[
						V_t(\exp(-R_{\rho} H^{n}) F^{n}) \le \E^{\Q^{n}}\left[\int_{0}^{t} \exp(-R_{\rho} H^{n}_s) R_{\rho} (F^{n}_s)^{1-\frac{1}{R_{\rho}}} \d s\right] \le R_{\rho} K^{1-\frac{1}{R_{\rho}}} t
				.\] 
				Hence, the sequence $(\exp(-R_{\rho} H^{n}) F^{n})_{n \in \N}$ satisfies the Meyer-Zheng condition \[
						\sup_n \{ V_t(\exp(-R_{\rho} H^{n}) F^{n}) + \sup_{s \le t} \E^{\Q^{n}}[|\exp(-R_{\rho} H^{n}_s) F^{n}_s|] \} < \infty \quad \text{ for all } t > 0
				,\] and so $(\exp(-R_{\rho} H^{n}) F^{n})_{n \in \N}$ is tight over $D$ (equipped with the Meyer-Zheng topology) by \cite[Thm.~4]{Meyer1984}.

				Moreover, the sequence $(H^{n}, \tilde{M}^{n}, U^n)$ is tight by Prokhorov's theorem since it converges weakly by assumption and \cite[Cor.~VI.6.29]{Jacod2003} and $C \times C \times C$ is Polish.
				Since the sequence $(\exp(R_{\rho} H^{n}))_{n \in \N}$ is tight over $C$ by the continuous mapping theorem and Prokhorov's theorem, the sequence $(F^{n})_{n \in \N}$ is also tight over $D$.
				Hence, the sequence $(H^{n}, \tilde{M}^{n}, U^n, F^{n})$ is tight over $C \times C \times C \times D$, and so there exists a subsequence (which we omit in notation) s.t.\ $(H^{n}, \tilde{M}^{n}, U^n, F^{n}) \to (H, \tilde{M}, U, F)$ weakly over $C \times C \times C \times D$.

				Define the operator \[
						\Phi: C^{\eta_0} \times D_K \to D, (X, Y) \mapsto \int_{0}^{t} \exp(-R_{\rho} X_s) R_{\rho} Y_s^{1-\frac{1}{R_{\rho}}} \d s + \exp(-R_{\rho} X_t) Y_t
				.\]
				Note that $\Phi$ is continuous:
				Fix $(X^{n}, Y^{n})_{n \in \N}$ with $(X^{n}, Y^{n}) \to (X, Y)$ in $C^{\eta_0} \times D_K$.
				For fixed $t$, the integral converges by dominated convergence.
				Hence, the left summand converges in measure, i.e.\ in Meyer-Zheng topology.
				Moreover, the right summand also clearly converges in measure.

				By the continuous mapping theorem, we thus have $(H^{n}, \tilde{M}^{n}, U_n, F^{n}, N^{n}) \to (H, \tilde{M}, U, F, N)$ weakly over $C \times C \times C \times D \times D$, where $N = \Phi(H, F)$.

                Next, note that \cite[Thm.~11,12]{Meyer1984} easily extend from respective canonical filtration to $\mathbb{G}:=\mathbb{F}^{H,\tilde{M},U,F}$ by also evaluating the test functions for the conditional expectation at the auxiliary processes.
                By \cite[Thm.~12]{Meyer1984}, $\tilde{M}$ is a $(\mathbb{G}, \mathbb{Q})$-local martingale with quadratic variation $U$.
                Hence, $\mathbb{F}^M$ is immersed in $\mathbb{G}$ by \cref{setting:assumption:M-PRP} and $\mathbb{G}$ satisfies all assumptions of \cref{section:setting} that are used in \cref{section:hjb}.
                
				Since the $N^{n}$ are uniformly bounded continuous martingales, \cite[Thm.~11]{Meyer1984} yields that $N$ is a $\mathbb{G}$-martingale.
				Since $N$ is uniformly bounded, it is uniformly integrable and hence closed.
				Taking the limit $t \to \infty$ in $\Phi(H, F)$ now yields that  \[
						N_t = \E^{\Q}\left[ \int_{0}^{\infty} \exp(-R_{\rho} H_s) R_{\rho} F_s^{1-\frac{1}{R_{\rho}}} \d s \condbar \mathcal{F}_t \right] 
				.\] 
				In particular, this means that $F$ is a solution to \eqref{hjb} over $\mathbb{G}$.
				Applying the same logic as above with $R_{\rho} = 1$ yields that $G^{n} \to G$ weakly over $D$.
				Since $G^{n} \le F^{n} \le K$, we also have $G \le F \le K$, and so $F$ is indeed the unique solution to \eqref{hjb} over $\mathbb{G}$ with $G \le F \le K$ by \cref{existence:uniformly_wellposed} applied for the filtration $\mathbb{G}$ (note that since $\mathbb{F}^M$ is immersed in both $\mathbb{F}$ and $\mathbb{G}$, it does not matter whether one conditions on $\mathbb{F}$ or $\mathbb{G}$ in the definition of $G$).
				In particular, this means $F$ is $\mathbb{F}^{M}$-adapted.
                Hence, $N$ is also $\mathbb{F}^M$-adapted, and thus an $\mathbb{F}^M$-martingale.
                Since $\mathbb{F}^M$ is immersed in $\mathbb{F}$ by \cref{setting:assumption:M-PRP}, $N$ is an $\mathbb{F}$-martingale, and thus $F$ is the unique solution to \eqref{hjb} with $G \le F \le K$ over $\mathbb{F}$.

				Since $F$ is $\mathbb{F}^{M}$-adapted, $N$ is continuous by \cref{existence:dynamics}.
				Hence, we have $N^{n} \to N$ weakly over $C$ by \cite{Pratelli1999}.
				Moreover, since the total variation $\operatorname{Var}(\int \exp(-R_{\rho} H^{n}_s) R_{\rho} (F^{n}_s)^{1-\frac{1}{R_{\rho}}} \d s)$ is strongly majorised by the $C$-tight sequence $(R_{\rho} K^{1-\frac{1}{R_{\rho}}} t)_n$, the sequence $(\int \exp(-R_{\rho} H^{n}_s) R_{\rho} (F^{n}_s)^{1-\frac{1}{R_{\rho}}})_n$ is $C$-tight by \cite[Prop.~VI.3.35,VI.3.36]{Jacod2003}.
				Recall that $\exp(-R_{\rho} H^{n}) F^{n} = - \int \exp(-R_{\rho} H^{n}_s) R_{\rho} (F^{n}_s)^{1-\frac{1}{R_{\rho}}} \d s + N^{n}$, so $(\exp(-R_{\rho} H^{n}) F^{n})_n$ is $C$-tight.
				Since $\exp(R_{\rho} H^{n})$ is $C$-tight, this means that $(F^{n})$ is also $C$-tight. 
				Hence, we have $(\tilde{M}^{n}, F^{n}) \to (\tilde{M}, F)$ weakly over $C \times C$.

				Finally, notice that $ \int Z^{n}_s \d \langle \tilde{M}^{n} \rangle_s = \langle F^{n}, \tilde{M}^{n} \rangle $.
				By \cref{existence:dynamics}, $F^{n}$ decomposes as $F^{n} = A^{n} + L^{n}$, where \[
						A^{n}_t = -\int_{0}^{t} R_{\rho} (F^{n}_s)^{1-\frac{1}{R_{\rho}}} \d s + \int_{0}^{t} R_{\rho} F^{n}_s \d H^{n}_s
				\] and $L^{n}$ is a continuous local martingale.
				For all $t>0$, we have \begin{equation}
						\label{stability:variation_bound}
						\operatorname{Var}(A^{n})_t \le \int_{0}^{t} R_{\rho} (F^{n}_s)^{1-\frac{1}{R_{\rho}}} \d s + R_{\rho} F^{n}_s \d H^{n}_s \le R_{\rho} K^{1-\frac{1}{R_{\rho}}} t + R_{\rho} K H^{n}_t
				,\end{equation} so $\sup_n \E^{\Q^{n}}[\operatorname{Var}(A^{n})_t] < \infty $ as $\sup \E^{\Q^{n}}[H^{n}_t] < \infty$ by assumption.
				Hence, the sequence $(\operatorname{Var}(A^{n})_t)_n$ is tight, and so we obtain that the sequence $(A^{n})_n$ is P-UT\footnote{Predictably Uniformly Tight, see \cite[Def.~VI.6.1]{Jacod2003}} by \cite[\nopp VI.6.6]{Jacod2003}.
				Moreover, \eqref{stability:variation_bound} yields that $\operatorname{Var}(A^{n})$ is strongly majorised by the tight sequence $(R_{\rho} K^{1-\frac{1}{R_{\rho}}} t + R K H^{n})_n$, so $(A^{n})_n$ is tight by \cite[Prop.~VI.3.35,VI.3.36]{Jacod2003}.
				Hence, the sequence $(L^{n})_n$ is also tight, and thus P-UT since any tight sequence of continuous local martingales is P-UT by \cite[Thm.~VI.6.21]{Jacod2003}.
				Together, $(F^{n})_n$ is P-UT by \cite[\nopp VI.6.4]{Jacod2003}.
				Moreover, $(\tilde{M}^{n})_n$ is P-UT by \cite[Thm.~VI.6.21]{Jacod2003} as it is a weakly convergent sequence of continuous local martingales.
				Hence, the joint sequence $(F^{n}, \tilde{M}^{n})_n$ is also P-UT by \cite[\nopp VI.6.3]{Jacod2003}.
				Since $(F^{n}, \tilde{M}^{n}) \to (F, \tilde{M})$ weakly in $C \times C$, \cite[Thm.~VI.6.26]{Jacod2003} now yields that $(F^{n}, \int Z^{n}_s \d \langle \tilde{M}^{n} \rangle_s) \to (F, \int Z_s \d U_s)$ weakly over $C \times C$.
		\end{proof}

        \subsection{Sign of the Control Process}

        In a Brownian setting, if the Malliavin derivative of the myopic consumption rate is non-negative, it is possible to identify the sign of the control process.
        Under these conditions, the operator $T$ preserves a non-positive Malliavin derivative of its input.
        Thus, starting the fixed point iteration from a constant process allows one to establish that the $Z$-component of the solution to the IVC-BSDE is non-positive.
        For instance, this will be the case in (regularised versions of) the Volterra Heston model considered in \cref{section:heston}.
        In the following, we denote by $\mathbb{L}^{1,2}_{\mathrm{loc}}$ the space of Malliavin differentiable processes $a$ that satisfy $\E[\int_0^T |a_t| \d t + \int_0^T \int_u^T |D_u a_t|^2 \d t \d u] < \infty$ for all $T > 0$.

		\begin{theorem}
                \label{properties:volatility:non-positive}
				Assume that $R > 1$.
                Let $W^\Q$ be a $\Q$-Brownian motion s.t.\ $\mathbb{F}^{W^\Q}$ is immersed in $\mathbb{F}$.
				Suppose that $d\tilde{M}_t = \sigma_t dW^{\Q}_t$ for some $\mathbb{F}^{W^\Q}$-progressively measurable process $\sigma \ge 0$, and that $dH_t = \eta_t dt $ for some $\mathbb{F}^{W^\Q}$-adapted process $\eta$ with $0 < \eta_0 \le \eta \le \eta_1$, where $\eta_0, \eta_1$ are positive constants.
                Let $F \in \mathcal{D}(T)$ be the solution to \cref{hjb} from \cref{existence:eta_bounded} with control process $Z^{F}$.
				If $\eta \in \mathbb{L}^{1,2}_{\mathrm{loc}}$ with $D_u \eta_t \ge 0$ for all $u \le t$, then $Z^F \le 0$.
		\end{theorem}
		\begin{proof}
				First, consider the time-truncated equation \[
						F^{S}_t = \E^{\Q}\left[ \int_{t}^{S} \exp(-R_\rho (H_s - H_t)) R_\rho (F^{S}_s)^{1-\frac{1}{R_\rho}} \d s + \exp(-R_\rho (H_S - H_t)) \eta_0^{-R_\rho} \condbar \mathcal{F}_t \right] =: (T^{S} F^{S})_t
				.\] 
				Since $\eta \ge \eta_0$, we have $T^{S} \eta_0^{-R_\rho} \le \eta_0^{-R\rho}$, i.e.\ $\eta_0^{-R_\rho}$ remains a supersolution to the time-truncated equation.
				Moreover, $\eta_1^{-R_\rho}$ is a subsolution.
				Define the sequence $F^{S,n}$ by the fixed point iteration $F^{S,1} = \eta_0^{-R\rho}$, $F^{S,n+1} = T^{S} F^{S,n}$.
				As in \cref{existence:fixed_point_iteration}, one sees that $F^{S,n} \downarrow F^{S}$, where $F^{S}$ is some solution to the time-truncated equation that satisfies $\eta_1^{-R_\rho} \le F^{S} \le \eta^{-R_{\rho}}$.

				We will show by induction that $F^{S,n} \in \mathbb{L}^{1,2}_{\mathrm{loc}}$ and $D_u F^{S,n}_t \le 0$ for all $u \le t, n \in \N$.
				The base case $n = 1$ holds since $F^{S,1}$ is constant.
				Now, fix $n \in \N$ and assume that $F^{S,n}$ is Malliavin differentiable with $F^{S,n} \in \mathbb{L}^{1,2}_{\mathrm{loc}}$ and $D_u F^{S,n}_t \le 0$ for $u \le t$.
				For $u \le t$, we have \[
						D_u F^{S,n+1} = \E^{\Q}\left[ \int_{t}^{S} D_u \left( \exp(-R_\rho (H_s - H_t)) R_\rho (F^{S,n}_s)^{1-\frac{1}{R_\rho}} \right) \d s + D_u \exp(-R_\rho (H_S - H_t)) \eta^{-R_\rho} \condbar \mathcal{F}_t \right] 
				\] by \cite[Prop.~1.2.8]{Nualart2006} and \cite[Prop.~3.4.3,Ex.~3.4.4]{Nualart2018} since $F^{S,n} \in \mathbb{L}^{1,2}_{\mathrm{loc}}$.
				Using the product and chain rule (see \cite[Prop.~1.2.3]{Nualart2006}, note that the chain rule is applicable since $F^{S,n}$ is bounded and bounded away from zero) and that $D_u (H_s - H_t) \ge 0$ since $D_u \eta_s \ge 0$ by \cite[Prop.~3.4.3,Ex.~3.4.4]{Nualart2018} as $\eta \in \mathbb{L}^{1,2}_{\mathrm{loc}}$ for all $T > 0$, we obtain \begin{align*}
						D_u & \left( \exp( -R_\rho (H_s - H_t)) R_\rho (F^{S,n}_s)^{1-\frac{1}{R_\rho}} \right) = \\ &=
						\exp( -R_\rho (H_s - H_t) ) ( -R_\rho D_u (H_s - H_t) ) R_\rho (F^{S,n}_s)^{1-\frac{1}{R_\rho}} \\ &\qquad + \exp( -R_\rho (H_s - H_t) ) R_\rho \left( 1 - \frac{1}{R_\rho} \right) (F^{S,n}_s)^{-\frac{1}{R_\rho}} D_u F^{S,n}_s \\ & \le 
				        0
				\end{align*} and \[
						D_u \exp(-R_\rho (H_S - H_t)) \eta^{-R_\rho} = -R_\rho D_u (H_S - H_t) \exp(-R_\rho (H_S - H_t)) \eta^{-R_\rho} \le 0
				,\] which implies that $D_u F^{S,n+1}_t \le 0$.
                Moreover, since the model is strongly myopically well-posed and $F^{S,n}, F^{S,n+1}$ are bounded and bounded away from zero, we have \[
                        |D_u F^{S,n+1}_t| \le C \E^\Q\left[ \int_t^S \left( \int_t^s |D_u \eta_v| \d v + |D_u F^{S,n}_s| \right) \d s + \int_t^S |D_u \eta_v| \d v \condbar \mathcal{F}_t \right]
                \] for some constant $C > 0$ (whose value will change from line to line).
                By Fubini-Tonelli, we have $\int_t^S \int_t^s |D_u \eta_v| \d v \d s = \int_t^S (S-v) |D_u \eta_v| \d v \le S \int_t^S |D_u \eta_v| \d v $, so that \[
                        |D_u F^{S,n+1}_t|^2 \le C \E^\Q\left[ \int_t^S \left( |D_u \eta_s|^2 + |D_u F^{S,n}_s|^2 \right) \d s \condbar \mathcal{F}_t \right]
                .\]
                Finally, we have \[
                        \E^Q\left[ \int_0^T \int_u^T |D_u F^{S,n+1}_t|^2 \d t \d u \right] \le C \E^\Q\left[ \int_0^T \int_u^T (|D_u \eta_t|^2 + |D_u F^{S,n}_t|^2) \d t \d u \right] < \infty
                ,\] by an analogous Fubini-Tonelli application as above, so that $F^{S,n+1} \in \mathbb{L}^{1,2}_{\mathrm{loc}}$ since $\eta, F^{S,n} \in \mathbb{L}^{1,2}_{\mathrm{loc}}$.

                Now, set $N^{S,n} = \int_0^t \exp(-R_\rho H_s) R_\rho (F^{S,n-1}_s)^{1-\frac{1}{R_\rho}} \d s + \exp(-R_\rho H_t) F^{S,n}_t$.
                As in \cref{existence:dynamics}, one sees that $N^{S,n}$ is a martingale.
                Moreover, using that $D_u H_t \ge 0$ and $D_u F^{S,n}_t \le 0$, one sees that $D_u N^{S,n}_S \le 0$. 
                Hence, the Clarke-Ocone formula (see e.g. \cite[1.3.14]{Nualart2006}) yields that $\exp(-R_\rho H_u) Z^{S,n}_u \sigma_u = \E^{\Q}\left[ D_u N^{S,n}_S \condbar \mathcal{F}_u \right] \le 0 $.
				Since $\exp(-R_\rho H) Z^{S,n} \sigma \le 0$ and $\exp(-R_\rho H), \sigma \ge 0$, the quadratic co-variation $\langle F^{S,n}, \tilde{M} \rangle = \int_{0}^{\cdot} Z^{S,n}_s \d \langle \tilde{M} \rangle_s$ is non-increasing.
				It remains to show that this carries over to $F^{S}$ and $F$.

				We can decompose $dF^{S,n} = dA^{S,n} + dL^{S,n}$, where $A^{S,n}_t = \int_{0}^{t} (-R_\rho (F^{S,n-1}_s)^{1-\frac{1}{R_\rho}} + R_\rho \eta_s F^{S,n}_s) \d s$ and $L^{S,n}$ is a local martingale.
				Since \[
						\operatorname{Var}(A^{S,n}) \le (R_\rho \eta_0^{1-R_\rho} + R_\rho \eta_1 \eta_0^{-R_\rho}) t
				,\] the sequence $(A^{S,n})_{n \in \N}$ is $C$-tight by \cite[Prop.~VI.3.35,VI.3.36]{Jacod2003} and P-UT by \cite[VI.6.6]{Jacod2003}.
				Moreover, since $F^{S,n} \downarrow F^{S}$, the sequence $(F^{S,n})_{n \in \N}$ is $C$-tight by Dini's theorem. 
				Hence, $(L^{S,n})_{n \in \N}$ is $C$-tight, and thus P-UT by \cite[Thm.~VI.6.21]{Jacod2003}.
				In total, the sequence $(F^{S, n})_{n \in \N}$ is P-UT, so that $\int_{0}^{\cdot} Z^{S,n}_s \d \langle \tilde{M} \rangle_s \to \int_{0}^{\cdot} Z^{S}_s \d \langle \tilde{M} \rangle_s$ by \cite[Thm.~VI.6.26]{Jacod2003} weakly over $C$.
				In particular, this means that $\int_{0}^{\cdot} Z^{S}_s \d \langle \tilde{M} \rangle_s $ is non-increasing.

				Finally, notice that since $\eta \ge \eta_0$, $F^{S'}$ satisfies $T^{S} F^{S'} \ge F^{S'}$ for all $S < S'$.
				By the classical finite-horizon BSDE comparison theorem with Lipschitz driver, this means that $F^{S} \ge F^{S'}$, i.e.\ the sequence $F^{S}$ is decreasing in $S$.
				Thus, we have $F^{S} \downarrow F$ for some process $F$ since $F^{S} \ge \eta_1^{-\tilde{R}}$ for all $S$.
				Since $\eta \ge \eta_0$, taking the limit in $F^{S} = T^{S} F^{S}$ via dominated convergence yields that $F = TF$.
				By uniqueness, $F$ is indeed the solution from \cref{existence:eta_bounded}.
				It follows analogously as in the paragraph above that $\int_{0}^{\cdot} Z^{S}_s \d \langle \tilde{M} \rangle_s \to \int_{0}^{\cdot} Z^F_s \d \langle \tilde{M} \rangle_s$ weakly over $C$.
				Hence, $\int_{0}^{\cdot} Z^F_s \d \langle \tilde{M} \rangle_s$ is non-increasing, and thus $Z^F \le 0$ $\Q \times d\langle \tilde{M} \rangle$-a.e.
		\end{proof}
        
		\section{Verification Theorem}
		\label{section:verification}

		In this section, we prove a general verification theorem that shows that the candidate solution to \eqref{eq:IVCBSDE} indeed parametrises the value process and optimal controls.

		\begin{theorem}
				\label{verification:verification}
				Let $F$ be a positive $\mathbb F^M$-adapted solution to \eqref{eq:IVCBSDE} with control process $Z^F$.
				Assume that the local martingale \[
						L_t = \mathcal{E}\left( \left(\frac{1-R}{R} \rho \Lambda + \frac{Z^F}{F}\right) \cdot M \right)_t \mathcal{E}\left( \frac{1-R}{R} \sqrt{1-\rho^2} \left( \Lambda + \rho \frac{R}{R_{\rho}} \frac{Z^F}{F} \right) \cdot M^{\perp} \right)_t
				\] is an $(\mathbb{F}, \mathbb{P})$-martingale.
				Moreover, assume that $\int_{0}^{\infty} F^{-\frac{1}{R_{\rho}}}_t \d t = \infty$ $\hat{\Q}$-a.s., where $\hat{\Q}$ is the measure defined by $\frac{d \hat{\Q}}{d \mathbb{P}} \Bigr|_{\mathcal{F}_t} = L_t$ (which exists by the extension property of $\mathbb{F}$).
				Then the value process and optimal controls are given by
                \begin{equation}
				        \label{verification:optimal_controls}
						V_t = \frac{X_t^{1-R}}{1-R} F^{\frac{R}{R_{\rho}}}_t, \qquad \hat{\Xi}_t = F^{-\frac{1}{R_{\rho}}}_t, \qquad \hat{\Pi}_t = \frac{1}{R} \Lambda_t + \frac{1}{R_{\rho}} \rho \frac{Z^F_t}{F_t}
				\end{equation}
		\end{theorem}

		\begin{remark}
				\label{verification:verification:sufficient_criteria}
				The condition $\int_{0}^{\infty} F^{-\frac{1}{R_{\rho}}}_t \d t = \infty $ $\hat{\Q}$-a.s.\ is automatically satisfied in strongly myopically well-posed models since $F$ is bounded from above.

                Suppose that the Dambis–Dubins–Schwarz Brownian motion $W^\perp$ associated to $M^\perp$ is independent of $\mathbb{F}^M$.
                This condition is satisfied in both \cref{setting:typical_factor_model,setting:singular_clocks}.
				Then $L$ is a $(\mathbb{F}, \mathbb{P})$-martingale if and only if \[
						\mathcal{E}\left( \frac{Z^F}{F} \cdot \tilde{M} \right) 
				\] is a $(\mathbb{F}^M, \Q$)-martingale since the $\frac{1-R}{R} \Lambda$ term arises from the change of measure from $\Q$ to $\mathbb{P}$ and the integrand as well as the time-change in the stochastic integral w.r.t.\ $M^\perp = W^\perp_U$ are $\mathbb{F}^M$-adapted.
		\end{remark}

		\begin{remark}
				\label{verification:brownian_filtration}
                Assume that there exists a Brownian motion $B$ with $\mathbb{F}^M \subseteq \mathbb{F}^B \subset \mathbb{F}$ s.t.\ $\Lambda$ and $H$ are $\mathbb{F}^B$-adapted and $\mathbb{F}^B$ is immersed in $\mathbb{F}$.
                For instance, the latter condition is satisfied if the Dambis-Dubins-Schwarz Brownian motion $W^\perp$ associated to $M^\perp$ is independent from $\mathbb{F}^B$ and $\mathbb{F} = \mathbb{F}^B \vee \mathbb{F}^{M^\perp}$.
                In particular, all of these conditions are satisfied in the setting of \cref{setting:typical_factor_model} if $Y$ is a strong solution.
                
				Then \cref{setting:assumption:M-PRP,setting:assumption:M-adapted} are not necessary.
				They are only required to ensure that there is no orthogonal term in the Galtchouk-Kunita-Watanabe decomposition of $F$ in \cref{existence:dynamics}, which is automatically the case in the Brownian filtration.
				In this case, we have $dM_t = \sigma_t dB_t$ for some process $\sigma$, and the optimal portfolio reads \[
						\hat{\Pi}_t = \frac{1}{R} \Lambda_t + \frac{1}{R_{\rho}} \rho \frac{\tilde{Z}^F_t}{\sigma_t F_t}
				,\] where $\tilde{Z}^F$ is the control process of $F$ w.r.t.\ $B$.
		\end{remark}

        		\begin{example}
				\label{verification:example:bounded_eta}
				Assume that $R>1$ and that the cumulative myopic consumption is Lipschitz, i.e.\ there exists a constant $C > 0$ s.t.\ $|H_s - H_t| \le C (s-t)$.
				In the absolutely continuous setting of \cref{setting:typical_factor_model}, this corresponds to models in which the myopic consumption rate $\eta$ is bounded.
				Assume moreover that the process \[
						\E^{\Q}\left[\int_{t}^{\infty} \exp( -R_{\rho} (H_s - H_t) ) \d s \condbar \mathcal{F}_t \right] 
				\] is bounded (which is the case in particular if the model is strongly myopically well-posed) and that the assumptions of \cref{verification:verification:sufficient_criteria} are satisfied.
				By \cref{existence:eta_bounded}, there exists a unique solution $F$ to \eqref{hjb} with control process $Z^F$ that is bounded and bounded away from zero.
				Since $H$ is Lipschitz, we have $|H_t| \le Ct$ for all $t \ge 0$.
				Hence, $\mathcal{E}(\frac{Z^F}{F} \cdot \tilde{M})$ is a true $(\mathbb{F}^M, \Q)$-martingale by \cref{properties:local_boundedness_away_from_0:martingale} below (applied with $\tau = t$ for arbitrary $t \ge 0$).
				By \cref{verification:verification,verification:verification:sufficient_criteria}, the value process and optimal controls are thus given by \eqref{verification:optimal_controls}.
		\end{example}
        
		\begin{proof}[Proof of \cref{verification:verification}]
				The proof uses a duality argument similar to the one by \textcite[Thm~3.3]{Guasoni2020} in the setting where $Y$ is an Itô diffusion. 
				In the following we denote $ \bar{\rho} = \sqrt{1 - \rho^2} $.
				For any progressively measurable process $\varphi$, we define the state-price density \[
						\zeta^{\varphi}_t = \exp\left( -\int_{0}^{t} r_u \d u \right)  \mathcal{E}\left( \left(\bar{\rho}^2 \varphi - \rho \Lambda\right) \cdot M - \bar{\rho}(\Lambda + \rho \varphi) \cdot M^{\perp} \right)_t
				.\]  
				The proof proceeds in three steps: 
				\begin{enumerate}
						\item Budget Feasibility Condition: For any admissible control $(\Pi, \Xi)$ and any progressively measurable process $\varphi$, we have $ \E[\int_{0}^{\infty} \zeta^{\varphi}_s \Xi_s X^{\Pi, \Xi}_s \d s ] \le X_0 $.
						\item $ \E\left[\int_{0}^{\infty} \exp\left( - \int_{0}^{t} \delta_u \d u \right) \frac{\left( \hat{\Xi}_t X^{\hat{\Pi}, \hat{\Xi}}_t \right)^{1-R} }{1-R} \d t \right] = \frac{X_0^{1-R}}{1-R} F_0^{\frac{R}{R_{\rho}}} $.
						\item There exists some $\hat{\varphi}$ s.t.\ $ \E\left[\int_{0}^{\infty} \exp\left( -\frac{1}{R} \int_{0}^{t} \delta_u \d u \right) (\zeta^{\hat{\varphi}}_t)^{\frac{R-1}{R}} \d t \right] = F_0^{\frac{1}{R_{\rho}}} $.
				\end{enumerate}
				By \cite[Lem.~A.1]{Guasoni2020}, the budget feasibility condition implies the upper bound \[
						\E\left[\int_{0}^{\infty} \exp\left( - \int_{0}^{t} \delta_u \d u \right)  \frac{(\Xi_t X^{\Pi, \Xi}_t)^{1-R}}{1-R} \d t \right] \le \frac{X_0^{1-R}}{1-R} \E\left[\int_{0}^{\infty} \exp\left( -\frac{1}{R} \int_{0}^{t} \delta_u \d u\right) (\zeta^{\varphi}_t)^{\frac{R-1}{R}} \d t \right]^{R}
				\]
				for any $\Pi, \Xi, \varphi$.
				Together with Steps 2 and 3, this upper bound yields that for any admissible control $(\Pi, \Xi)$ \begin{align*}
						\E\Bigg[\int_{0}^{\infty} \exp\left( - \int_{0}^{t} \delta_u \d u \right) & \frac{( \Xi_t X^{\Pi, \Xi}_t ) ^{1-R}}{1-R} \d t \Bigg] \le \\ &\le
						\frac{X_0^{1-R}}{1-R} \E\left[\int_{0}^{\infty} \exp\left( -\frac{1}{R} \int_{0}^{t} \delta_u \d u \right)  (\zeta^{\hat{\varphi}})^{\frac{R-1}{R}} \d t \right]^{R} =\\&=
						\frac{X_0^{1-R}}{1-R} F_0^{\frac{R}{R_{\rho}}} = 
						\E\left[\int_{0}^{\infty} \exp\left( - \int_{0}^{t} \delta_u \d u  \right)  \frac{( \hat{\Xi}_t X^{\hat{\Pi}, \hat{\Xi}}_t )^{1-R}}{1-R} \d t \right]  
				,\end{align*} so $(\hat{\Pi}, \hat{\Xi})$ is indeed optimal.

                Finally, notice that $(\hat{\Pi}, \hat{\Xi})$ is admissible: 
                Firstly, we have $\int_0^t \hat{\Xi}_s \d s < \infty$ a.s.\ for all $t>0$ since $F$ is positive and has continuous sample paths.
                Moreover, one sees that $\frac{Z}{F} \cdot \tilde{M}$ is a well-defined local martingale by applying Itô's formula to $\log F$, so that $\int_0^t \frac{Z_s^2}{F_s^2} \d U_s < \infty$ a.s.\ for all $t > 0$.
                This means that \[
                        \int_0^t \hat{\Pi}_s^2 \d U_s \le \int_0^t \left( \frac{2}{R^2} \Lambda_s^2 + \frac{2}{R_\rho^2} \rho^2 \frac{Z_s^2}{F_s^2} \right) \d U_s < \infty \text{ a.s.}
                \] for all $t > 0$.
            
				\textbf{Step 1:}
				Fix $\Pi, \Xi, \varphi$. 
				Set $G_t = \zeta^{\varphi}_t X^{\Pi, \Xi}_t + \int_{0}^{t} \zeta^{\varphi}_s \Xi_s X_s \d s$.
				Then \[
						dG_t = \zeta^{\varphi}_t X^{\Pi, \Xi}_t ((\rho (\Pi_t - \Lambda_t) + \bar{\rho}^2 \varphi_t) dM_t + \bar{\rho} (\Pi_t - \Lambda_t - \rho \varphi_t) dM^{\perp}_t)
				,\] so $G$ is a local martingale.
				Since $G$ is non-negative, $G$ is a supermartingale, and so \[
						X_0 = G_0 \ge \liminf_{t\to \infty} \E[G_t] \ge \liminf_{t \to \infty} \E\left[\int_{0}^{t} \zeta^{\varphi}_s \Xi_s X^{\Pi, \Xi}_s \d s \right] \ge \E\left[\int_{0}^{\infty} \zeta^{\varphi}_s \Xi_s X^{\Pi, \Xi}_s \d s \right] 
				.\] 

				\textbf{Step 2:}
				Using the form of the candidate optimal controls and the dynamics of the wealth process \eqref{setting:wealth_dynamics}, the candidate optimal wealth process $X^{\hat{\Pi}, \hat{\Xi}}$ is given by \[
						X^{\hat{\Pi}, \hat{\Xi}}_t = X_0 \exp\left( \int_{0}^{t} \left( r_s - F_s^{-\frac{1}{R_{\rho}}} \right) \d s + \int_{0}^{t} \left( \frac{\Lambda_s^2}{R} + \frac{1}{R_{\rho}} \rho \Lambda_s \frac{Z^F_s}{F_s} \right) \d U_s \right) \mathcal{E}\left( \rho \hat{\Pi}_s \cdot M + \bar{\rho} \hat{\Pi}_s \cdot M^{\perp}_t \right)_t
				.\] 
				Hence, 
				\begin{align*}
						\exp\bigg( - & \int_{0}^{t} \delta_s \d s \bigg)  \frac{(\hat{\Xi}_t X^{\hat{\Pi}, \hat{\Xi}}_t)^{1-R}}{1-R} \\&= 
						\frac{X_0^{1-R}}{1-R} F_t^{-\frac{1-R}{R_{\rho}}} \exp\bigg( \int_{0}^{t} \left(- \delta_s + (1-R) \left( r_s - F_s^{-\frac{1}{R_{\rho}}} \right) \right) \d s \\&\qquad + (1-R) \int_{0}^{t} \left( \frac{\Lambda_s^2}{R} + \frac{1}{R_{\rho}} \rho \Lambda_s \frac{Z^F_s}{F_s} \right) \d U_s + \int_{0}^{t} \rho \left( \frac{1-R}{R} \Lambda_s + \frac{1-R}{R_{\rho}} \rho \frac{Z^F_s}{F_s} \right) \d M_s \\&\qquad + \int_{0}^{t} \bar{\rho} \left( \frac{1-R}{R} \Lambda_s + \frac{1-R}{R_{\rho}} \rho \frac{Z^F_s}{F_s} \right) \d M^{\perp}_s - \int_{0}^{t} \frac{1-R}{2} \left( \frac{1}{R} \Lambda_s + \frac{1}{R_{\rho}} \rho \frac{Z^F_s}{F_s} \right)^2 \d U_s \bigg) \\&= 
						\frac{X_0^{1-R}}{1-R} F_t^{\frac{R-1}{R_{\rho}}} \exp\bigg(-R H_t - (1-R) \int_{0}^{t} F_s^{-\frac{1}{R_{\rho}}} \d s \\&\qquad + \int_{0}^{t} \left( \frac{1-R}{2R} \Lambda_s^2 + \frac{1-R}{R_{\rho}} \rho \Lambda_s \frac{Z^F_s}{F_s} - \frac{1-R}{2} \left( \frac{1}{R} \Lambda_s + \frac{1}{R_{\rho}} \rho \frac{Z^F_s}{F_s} \right)^2 \right) \d U_s \\&\qquad \int_{0}^{t} \rho \left( \frac{1-R}{R} \Lambda_s + \frac{1-R}{R_{\rho}} \rho \frac{Z^F_s}{F_s} \right) \d M_s + \int_{0}^{t} \bar{\rho} \left( \frac{1-R}{R} \Lambda_s + \frac{1-R}{R_{\rho}} \rho \frac{Z^F_s}{F_s} \right) \d M^{\perp}_s \bigg) 
				.\end{align*} 
				By Itô's lemma, the $\mathbb{P}$-dynamics of $\log F$ are given by 
                \begin{equation}
                        \label{verification:log_F_dynamics}
						d \log F_t = - R_{\rho} F^{-\frac{1}{R_{\rho}}} dt + R_{\rho} dH_t - \left( \frac{1-R}{R} \rho \Lambda_t \frac{Z^F_t}{F_t} + \frac{1}{2} \frac{(Z^F_t)^2}{F_t^2} \right) d U_t + \frac{Z^F_t}{F_t} dM_t
                .\end{equation}
                From this and $(1-R) \rho^2 + R = R_{\rho}$, we obtain 
				\begin{align}
						\label{verification:before_stoch_exponential}
						\exp\bigg( - & \int_{0}^{t} \delta_s \d s \bigg) \frac{(\hat{\Xi}_t X^{\hat{\Pi}, \hat{\Xi}}_t)^{1-R}}{1-R} \nonumber\\&=
						\frac{X_0^{1-R}}{1-R} F_t ^{\frac{R-1}{R_{\rho}}} \exp\bigg( -\frac{R}{R_{\rho}}(\log F_t - \log F_0) - \int_{0}^{t} F_s^{-\frac{1}{R_{\rho}}} \d s \nonumber\\&\qquad + \int_{0}^{t} \left(\frac{1-R}{2R} \Lambda_s^2 - \frac{1}{2} \frac{R}{R_{\rho}} \frac{(Z^F_s)^2}{F_s^2} - \frac{1-R}{2} \left(\frac{1}{R} \Lambda_s + \frac{1}{R_{\rho}} \rho \frac{Z^F_s}{F_s}\right)^2 \right) \d U_s \nonumber\\&\qquad \int_{0}^{t} \left( \frac{1-R}{R} \rho \Lambda_s + \frac{Z^F_s}{F_s} \right) \d M_s + \int_{0}^{t} \bar{\rho} \left( \frac{1-R}{R} \Lambda_s + \frac{1-R}{R_{\rho}} \rho \frac{Z^F_s}{F_s} \right) \d M^{\perp}_s \bigg) 
				.\end{align}
				Using that $R_{\rho} = (1-\rho^2) R + \rho^2 = (1-R) \rho^2 + R$ and $(1-R)(1-\rho^2) = 1-R_{\rho}$, we get 
				\begin{align*}
						\frac{1-R}{2R} &\Lambda^2 - \frac{1}{2} \frac{R}{R_{\rho}} \frac{(Z^F_s)^2}{F_s^2} - \frac{1-R}{2} \left(\frac{1}{R} \Lambda_s + \frac{1}{R_{\rho}} \rho \frac{Z^F_s}{F_s}\right)^2 + \frac{1}{2} \left( \frac{1-R}{R} \rho \Lambda_s + \frac{Z^F_s}{F_s} \right)^2 \\&\qquad + \frac{1}{2} (1-R)^2 \bar{\rho}^2 \left( \frac{1}{R} \Lambda_s + \frac{1}{R_{\rho}} \rho \frac{Z^F_s}{F_s} \right)^2 \\&=
						\frac{1-R}{2R^2} \Lambda_s^2 \left(R-1 + (1-R)\rho^2 + (1-R)(1-\rho^2)\right) \\&\qquad + \frac{1-R}{R R_{\rho}} \rho \Lambda_s \frac{Z^F_s}{F_s} \left( -1 + R_{\rho} + (1-R) (1-\rho^2) \right) \\&\qquad + \frac{1}{2R_{\rho}^2} \frac{(Z^F_s)^2}{F_s^2} \left( -R R_{\rho} - (1-R)\rho^2 + R_{\rho}^2 + (1-R)^2 (1-\rho^2) \rho^2 \right) \\&=
						0
				.\end{align*}
				This means that \begin{align*}
						\exp\bigg( - & \int_{0}^{t} \delta_s \d s \bigg) \frac{(\hat{\Xi}_t X^{\hat{\Pi}, \hat{\Xi}}_t)^{1-R}}{1-R} =
						\frac{X_0^{1-R}}{1-R} F_0^{\frac{R}{R_{\rho}}} F_t ^{-\frac{1}{R_{\rho}}} \exp\left( - \int_{0}^{t} F_s^{-\frac{1}{R_{\rho}}} \d s \right) \\&\qquad \mathcal{E}\left( \left(\frac{1-R}{R} \rho \Lambda + \frac{Z^F}{F}\right) \cdot M \right)_t \mathcal{E}\left( \frac{1-R}{R} \sqrt{1-\rho^2} \left( \Lambda + \rho \frac{R}{R_{\rho}} \frac{Z^F}{F} \right) \cdot M^{\perp} \right)_t \\&=
						\frac{X_0^{1-R}}{1-R} F_0^{\frac{R}{R_{\rho}}} F_t ^{-\frac{1}{R_{\rho}}} \exp\left( - \int_{0}^{t} F_s^{-\frac{1}{R_{\rho}}} \d s \right) L_t
				.\end{align*} 
				Now, using Fubini-Tonelli and that $\int_{0}^{\infty} F_s^{-\frac{1}{R_{\rho}}} \d s = \infty$ $\hat{\Q}$-a.s., we obtain 
				\begin{align*}
						\E\bigg[\int_{0}^{\infty} & \exp\left( - \int_{0}^{t} \delta_s \d s \right) \frac{(\hat{\Xi}_t X^{\hat{\Pi}, \hat{\Xi}}_t)^{1-R}}{1-R} \d t \bigg] \\&=
						\E\left[\int_{0}^{\infty} \frac{X_0^{1-R}}{1-R} F_0^{\frac{R}{R_{\rho}}} F_t ^{-\frac{1}{R_{\rho}}} \exp\left(- \int_{0}^{t} F_s^{-\frac{1}{R_{\rho}}} \d s \right) L_t \d t \right] \\&=
						\frac{X_0^{1-R}}{1-R} F_0^{\frac{R}{R_{\rho}}} \int_{0}^{\infty} \E\left[F_t ^{-\frac{1}{R_{\rho}}} \exp\left( -\int_{0}^{t} F_s^{-\frac{1}{R_{\rho}}} \d s  \right) L_t \d t \right] \\&=
						\frac{X_0^{1-R}}{1-R} F_0^{\frac{R}{R_{\rho}}} \int_{0}^{\infty} \E^{\hat{\Q}}\left[F_t ^{-\frac{1}{R_{\rho}}} \exp\left( - \int_{0}^{t} F_s^{-\frac{1}{R_{\rho}}} \d s \right) \right] \d t \\&=
						\frac{X_0^{1-R}}{1-R} F_0^{\frac{R}{R_{\rho}}} \E^{\hat{\Q}}\left[\int_{0}^{\infty} F_t ^{-\frac{1}{R_{\rho}}} \exp\left( - \int_{0}^{t} F_s^{-\frac{1}{R_{\rho}}} \d s  \right) \d t \right] \\&=
						\frac{X_0^{1-R}}{1-R} F_0^{\frac{R}{R_{\rho}}} \E^{\hat{\Q}}\left[1 - \exp\left(-\int_{0}^{\infty} F_s^{-\frac{1}{R_{\rho}}} \d s \right) \right] \\&=
						\frac{X_0^{1-R}}{1-R} F_0^{\frac{R}{R_{\rho}}}
				.\end{align*}

				\textbf{Step 3:}
				Choose $\hat{\varphi}_t = \frac{R}{R_{\rho}} \frac{Z^F_t}{F_t} $.
				We then have \begin{align*}
						\exp&\left( -\frac{1}{R} \int_{0}^{t} \delta_s \d s \right) \left(\zeta^{\hat{\varphi}}_t\right)^{\frac{R-1}{R}} =\\&= 
						\exp\bigg( \int_{0}^{t} \left( -\frac{1}{R} \delta_s - \frac{R-1}{R} r_s \right) \d s \\&\qquad - \int_{0}^{t} \left( \frac{R-1}{2R} \left( \bar{\rho}^2 \frac{R}{R_{\rho}} \frac{Z^F_s}{F_s} - \rho \Lambda_s\right)^2 + \frac{R-1}{2R} \bar{\rho}^2 \left( \Lambda_s + \rho \frac{R}{R_{\rho}} \frac{Z^F_s}{F_s} \right)^2 \right) \d{} U_s \\&\qquad + \int_{0}^{t} \frac{R-1}{R} \left( \bar{\rho}^2 \frac{R}{R_{\rho}} \frac{Z^F_s}{F_s} - \rho \Lambda_s\right) \d M_s - \int_{0}^{t} \frac{R-1}{R} \bar{\rho} \left( \Lambda_s + \rho \frac{R}{R_{\rho}} \frac{Z^F_s}{F_s} \right) \d M^{\perp}_s \bigg) \\&=
						\exp\bigg(-H_t + \int_{0}^{t} \left( - \frac{(1-R)^2}{2R^2} \Lambda_s^2 + \frac{1}{2}(1-R) R \frac{1}{R_{\rho}^2} \bar{\rho}^2 \frac{(Z^F_s)^2}{F_s^2} \right) \d U_s \\&\qquad + \int_{0}^{t} \left( \frac{1-R}{R} \rho \Lambda_s + \frac{R_{\rho}-1}{R_{\rho}} \frac{Z^F_s}{F_s} \right) \d M_s + \int_{0}^{t} \frac{1-R}{R} \bar{\rho} \left(\Lambda_s + \rho \frac{R}{R_{\rho}} \frac{Z^F_s}{F_s} \right) \d M^{\perp}_s \bigg) 
				.\end{align*} 
				Similarly as in Step 2, using the $\mathbb{P}$-dynamics of $\log F$ given in \eqref{verification:log_F_dynamics}, we obtain \begin{align*}
				\exp\bigg( - & \frac{1}{R} \int_{0}^{t} \delta_s \d s \bigg) \left(\zeta^{\hat{\varphi}}_t\right)^{\frac{R-1}{R}} \\&= 
				F_0^{\frac{1}{R_{\rho}}} F^{-\frac{1}{R_{\rho}}}_t \exp\bigg( -\int_{0}^{t} F_s^{-\frac{1}{R_{\rho}}} \d s \\&\qquad + \int_{0}^{t} \left( - \frac{1-R}{R R_{\rho}} \rho \Lambda_s \frac{Z^F_s}{F_s} - \frac{1}{2R_{\rho}} \frac{(Z^F_s)^2}{F_s^2} - \frac{(1-R)^2}{2R^2} \Lambda_s^2 + \frac{1}{2}(1-R) R \frac{1}{R_{\rho}^2} \bar{\rho}^2 \frac{(Z^F_s)^2}{F_s^2} \right) \d{} U_s \\&\qquad + \int_{0}^{t} \left( \frac{1-R}{R} \rho \Lambda_s + \frac{Z^F_s}{F_s} \right) \d M_s + \int_{0}^{t} \frac{1-R}{R} \bar{\rho} \left(\Lambda_s + \rho \frac{R}{R_{\rho}} \frac{Z^F_s}{F_s}\right) \d M^{\perp}_s \bigg) \\&=
				F_0^{\frac{1}{R_{\rho}}} F^{-\frac{1}{R_{\rho}}}_t \exp\bigg( - \int_{0}^{t} F_s^{-\frac{1}{R_{\rho}}} \d s \\&\qquad + \int_{0}^{t} \left( \frac{1-R}{2R} \Lambda_s^2 - \frac{1}{2} \frac{R}{R_{\rho}} \frac{(Z^F_s)^2}{F_s^2} - \frac{1-R}{2} \left( \frac{1}{R} \Lambda_s + \frac{1}{R_{\rho}} \rho \frac{Z^F_s}{F_s} \right)^2 \right) \d U_s \\&\qquad + \int_{0}^{t} \left( \frac{1-R}{R} \rho \Lambda_s + \frac{Z^F_s}{F_s} \right) \d M_s + \int_{0}^{t} \bar{\rho} \left(\frac{1-R}{R} \Lambda_s + \frac{1-R}{R_{\rho}} \rho \frac{Z^F_s}{F_s}\right) \d M^{\perp}_s \bigg)
				.\end{align*} 
				Notice that the integrals in the exponential are exactly the same as in \eqref{verification:before_stoch_exponential} in Step 2.
				Hence, we have \[
						\exp\left( -\frac{1}{R} \int_{0}^{t} \delta_s \d s \right) \left(\zeta^{\hat{\varphi}}_t\right)^{\frac{R-1}{R}} = F_0^{\frac{1}{R_{\rho}}} F_t ^{-\frac{1}{R_{\rho}}} \exp\left( -\int_{0}^{t} F_s^{-\frac{1}{R_{\rho}}} \d s \right) L_t
				.\] 
				As in Step 2, this yields that \[
						\E\left[\int_{0}^{\infty} \exp\left( -\frac{1}{R} \int_{0}^{t} \delta_s \d s \right) \left(\zeta^{\hat{\varphi}}_t\right)^{\frac{R-1}{1}} \d t \right] = 
						F_0^{\frac{1}{R_{\rho}}} \E^{\hat{\Q}}\left[1 - \exp\left(- \int_{0}^{\infty} F_s^{-\frac{1}{R_{\rho}}} \d s  \right) \right] =
						F_0^{\frac{1}{R_{\rho}}}
				.\qedhere\] 
		\end{proof}

        The martingality (as opposed to local martingality) required in \cref{verification:verification} (resp.\ \cref{verification:verification:sufficient_criteria}) is satisfied as long as the candidate solution to \eqref{hjb} is bounded away from zero.
        Some conditions for when this is the case (up to a stopping time) are collected in \cref{properties:local_boundedness_away_from_0}.
        
        \begin{proposition}
				\label{properties:local_boundedness_away_from_0:martingale}
				Let $F$ be an $\mathbb{F}^{M}$-adapted solution to \eqref{eq:IVCBSDE} with control process $Z^F$ and let $\tau$ be a bounded $\mathbb{F}^M$-stopping time.
				Assume that there exist constants $C_1, C_2 > 0$ s.t.\ $C_1 \le F \le C_2$ and $|H| \le C_2$ a.s.\ over $[0, \tau]$.
				Then $\mathcal{E}\left(\frac{Z^F}{F} \cdot \tilde{M}\right)$ is a true $(\mathbb{F}^M, \Q)$-martingale over $[0, \tau]$.
		\end{proposition}
		\begin{proof}
				Since $F$ is bounded and bounded away from $0$ and $H$ is bounded over $[0, \tau]$, the process $G := \log F - R_{\rho} H$ is bounded over $[0, \tau]$.
				Set $Z^G = \frac{Z^F}{F}$.
				Then $G$ has $\Q$-dynamics \[
						dG_t = -R_{\rho} \exp\left( -\frac{1}{R_{\rho}} G_t - H_t \right) dt - \frac{1}{2} (Z^G_t)^2 d U_t + Z^G_t d\tilde{M}_t
				.\] 

				Possibly extending the probability space, let $B$ be a Brownian motion independent of $\mathbb{F}^{M}$.
				We will re-cast the dynamics of $G$ as a BSDE that is driven by the local martingale $(\tilde{M}, B)$.
				Note that $\mathbb{F}^{M}$ is continuous by \cref{setting:assumption:M-PRP}.
				Moreover, note that $G, Z^G$ are $\mathbb{F}^{M}$-adapted, so we can consider the BSDE under the continuous filtration $\mathbb{F}^{M,B}$ (note that continuity of the filtration is a standing assumption in both \cite{Constantinou2025,Kazamaki1994}, results from which are used below).
				Set $I_t = U_t + t$. 
				By Radon-Nikodym, there exist predictable processes $\kappa^{1}, \kappa^2$ with $0 \le \kappa^{1} \le 1, 0 \le \kappa^2 \le 1$ s.t.\ $U_t = \int_{0}^{t} \kappa^{1}_s \d I_s $, $t = \int_{0}^{t} \kappa^2_s \d I_s $ (cf. \cite[Prop.~I.3.13]{Jacod2003}).
				Hence, $(G, (Z^G, 0))$ solves the BSDE \[
						dG_t = \left(-R_{\rho} \exp\left( -\frac{1}{R_{\rho}} G_t - H_t \right) \kappa^2_t - \frac{1}{2} (Z^G_t)^2 \kappa^{1}_t\right) dI_t + Z^G_t d\tilde{M}_t + 0 dB_t
				\] with (random) terminal horizon $\tau$ and terminal condition $G_\tau = \log F_\tau - R_{\rho} H_\tau$.
				Since $\tau$ is bounded by some $T > 0$, we can view this as a BSDE with deterministic terminal horizon $T$ by multiplying the driver as well as the control process by $ \mathbbm{1}_{\{t \le \tau\}}  $.
				Moreover, notice that this is a quadratic BSDE with bounded terminal condition since $G, H$ are bounded over $[0, \tau]$.
				Hence, \cite[Prop.~3.3]{Constantinou2025} yields that $Z^G \cdot \tilde{M}$ is a BMO-martingale, and thus $\mathcal{E}(\frac{Z^F}{F} \cdot \tilde{M})$ is a true $\mathbb{F}^{M,B}$-martingale over $[0, \tau]$ by \cite[Thm.~2.3]{Kazamaki1994}.
				Since $\mathcal{E}(\frac{Z^F}{F} \cdot \tilde{M})$ is $\mathbb{F}^{M}$-adapted, it is also a true $\mathbb{F}^{M}$-martingale.
		\end{proof}
        
		\section{Application to Volterra Heston Models with $L^1$ Kernels}
		\label{section:heston}

        \subsection{Setting}
		\label{subsection:heston:setting}

		In this section, we consider the Volterra Heston model with $L^1_{\mathrm{loc}}$ kernels, which in particular also includes $L^2_{\mathrm{loc}}$ kernels.
		For a locally square-integrable kernel $K$, the Volterra Heston variance process is defined by the stochastic Volterra equation \[
				Y_t = Y_0 + \int_{0}^{t} K(t-s) (-\kappa (Y_s - \theta)) \d s + \int_{0}^{t} K(t-s) \nu \sqrt{Y_s} \d B_s
		.\] 
		For the stochastic integral in the equation to be well-defined, it is necessary that $K \in L^2_{\mathrm{loc}}$.
		Using the affine structure of the Heston model, we can get past this restriction by considering the integrated variance $U_t = \int_{0}^{t} Y_s \d s$.
		Through an application of the stochastic Fubini theorem, one sees that 
        \begin{equation}
                \label{heston:eq_iv}
                U_t = \int_{0}^{t} \left(Y_0 + \int_{0}^{s} K(s-u) \kappa \theta \d u\right) \d s + \int_{0}^{t} K(t-s) (-\kappa U_s + \nu M_s) \d s
        ,\end{equation}
        where $M = \int_{0}^{t} \sqrt{Y_s} \d B_s$ is a local martingale with quadratic variation $U$ (see \cite[Lem.~2.1]{AbiJaber2021}).
		Crucially, this representation avoids any stochastic integrals involving $K$, so it only requires $K \in L^{1}_{\mathrm{loc}}$.
		In the case of the fractional kernel $K(t) = t^{h-\frac{1}{2}}$, this allows for $h \in (-\frac{1}{2}, \frac{1}{2})$ (instead of just $h \in (0, \frac{1}{2})$).
		For $h \in (-\frac{1}{2}, 0]$, the model is called \emph{hyper-rough}.
		As shown by \textcite[Thm.~3.2]{Jusselin2020}, the integrated variance process $U$ is no longer absolutely continuous when $h \in (-\frac{1}{2}, 0]$, so the spot variance process $Y$ does not exist any more.
		As $h \downarrow -\frac{1}{2}$, it was shown by \textcite{AbiJaber2026} that $U$ converges to an Inverse Gaussian Lévy process.

        We assume that the kernel $K \in L^{1}_{loc}([0, \infty), \R)$ as well as the shifted kernels $\Delta_{\varepsilon} K = K(\cdot + \varepsilon)$, $\varepsilon \in (0, 1)$, satisfy the following condition:
		The kernel is non-negative, non-increasing, and continuously differentiable on $(0, \infty)$, it has a resolvent of the first kind\footnote{The resolvent of the first kind is a measure $L$ of locally bounded variation that satisfies $K * L = 1$. A resolvent of the first kind does not exist for all kernels. If it exists, it is unique, see \cite[Thm.~5.5.5]{Gripenberg1990}.} $L$; moreover $L$ is non-negative and non-increasing in the sense that $s \mapsto L([s, s+t])$ is non-increasing for all $t \ge 0$.

		By \cite[Thm.~2.13,Rem.~2.14]{AbiJaber2021}, this guarantees the existence of a weak solution $(U, M)$ to \eqref{heston:eq_iv}, that is, a tuple $(U, M)$ s.t.\ $U$ is continuous, non-decreasing and satisfies \eqref{heston:eq_iv}, and $M$ is a continuous local martingale with quadratic variation $\langle M \rangle = U$.
		Moreover, the solution is unique in law in the sense that if $(U, M), (U', M')$ are two solutions (defined on potentially different probability spaces), then $\operatorname{Law}(U) = \operatorname{Law}(U')$.
		A sufficient condition for $K$ to satisfy these assumptions is if $K$ is completely monotone\footnote{A function $f \in C^{\infty}((0, \infty), \R)$ is called completely monotone if $(-1)^{k} f^{(k)} \ge 0$ for all $k \in \N_0$.} and not identically zero.
		In particular, this includes the gamma kernel $t^{\alpha - 1} e^{-\beta t}$ with $\alpha \in (0, 1], \beta \ge 0$ as well as shifted versions thereof.

        For the orthogonal noise, let $W^\perp$ be a Brownian motion independent of $\mathbb{F}^M$ and set $M^\perp = W^\perp_U$.
        Seting $\mathbb{F} = \mathbb{F}^M \vee \mathbb{F}^{M^\perp}$, $M$ and $M^\perp$ are strongly orthogonal $\mathbb{F}$-local martingales.

		In total, we consider the IVC model 
		\begin{align*}
				dS_t &= S_t \left(r dt + \Lambda dU_t + \rho dM_t + \sqrt{1-\rho^2} dM^{\perp}_t\right), \\
		          U_t &= \int_{0}^{t} \left(Y_0 + \int_{0}^{s} K(s-u) \kappa \theta \d u\right) \d s + \int_{0}^{t} K(t-s) (-\kappa U_s + \nu M_s) \d s
		.\end{align*}
		The parameters are assumed to satisfy $Y_0, r, \kappa, \theta, \nu, \delta > 0$, $\Lambda \in \R$, $\rho \in [-1, 1]$, and preferences are such that $R > 1$.
		We denote \[
				G_0(t) = \int_{0}^{t} \left(Y_0 + \int_{0}^{s} K(s-u) \kappa \theta \d u\right) \d s = \int_{0}^{t} \left(Y_0 + \kappa \theta \int_{0}^{s} K(u) \d u \right) \d s
		.\]

		In this model, the cumulative myopic consumption $H$ is given by \[
				H_t = \frac{1}{R} (\delta + (R-1) r) t + \frac{R-1}{2R^2} \Lambda^2 U_t
		.\] 
		We denote $\eta_0 = \frac{1}{R} (\delta - (1-R) r)$, $\eta_1 = \frac{R-1}{2R^2} \Lambda^2$.
        Note that $\eta_0 > 0, \eta_1 > 0$ since $R>1, r>0, \delta>0$.
		Notice that the model is strongly myopically well-posed with constant $\eta_0$.

		Notice that $\mathcal{E}\left( \frac{1-R}{R} \rho \Lambda \cdot M \right) $ is a true martingale by \cite[Lem.~6.1]{AbiJaber2021}, so $\Q$ is well-defined.
		Under $\Q$, $(U, \tilde{M})$ is a solution to the affine stochastic Volterra equation with modified parameters $\tilde{\kappa} = \kappa + \frac{R-1}{R} \rho \Lambda \nu$, $\tilde{\theta } = \frac{\kappa}{\tilde{\kappa}} \theta$, where $\tilde{M} = M + \frac{R-1}{R} \rho \Lambda U$.
		Note that we do not make any assumptions on $\tilde{\kappa}$.
        When $\tilde{\kappa} = 0$, we formally write $\tilde{\kappa} \tilde{\theta} := \kappa \theta$.

		We verify the predictable representation property of $\tilde{M}$ over $\mathbb{F}^{M}$ using \cite[Cor.~III.4.31]{Jacod2003}.
		To this end, fix $\Q' \ll \Q$ s.t.\ $\Q' \Bigr|_{\mathcal{F}^{M}_0} = \Q \Bigr|_{\mathcal{F}^{M}_0} $ and $\tilde{M}$ has semimartingale-characteristics $(0, U, 0)$ under $\Q'$.
		Since the affine stochastic Volterra equation also holds under $\Q'$ (as $\Q' \ll \Q$) and the $\Q'$-quadratic variation of $\tilde{M}$ is $U$, $(U, \tilde{M})$ solves the equation under $\Q'$.
		By weak uniqueness, this means that $\operatorname{Law}_{\Q}(U) = \operatorname{Law}_{\Q'}(U)$.
		Convolving both sides of the equation with the resolvent of the first kind $L$, we obtain \[
				\nu \tilde{M} * 1 = (U - G_0) * L + \tilde{\kappa} U * 1
		,\] so $\tilde{M} * 1$ can be recovered from $U$.
		Since $\tilde{M}$ is continuous, $\tilde{M}$ can then be recovered from $\tilde{M} * 1$ by differentiating.
		Thus, we obtain that $\operatorname{Law}_{\Q'}(\tilde{M}) = \operatorname{Law}_{\Q}(\tilde{M})$.
		As we work over $\mathbb{F}^{M} = \mathbb{F}^{\tilde{M}}$, this means that $\Q = \Q'$.
		Hence, $\tilde{M}$ has the predictable representation property over $\mathbb{F}^{M}$ by \cite[Cor.~III.4.31]{Jacod2003}.

		In conclusion, \cref{setting:assumption:Q,setting:assumption:M-PRP,setting:assumption:M-adapted} are all satisfied in this model.

		\begin{remark}
				We make only minimal assumptions on the model parameters.
				In particular, our results generalise the treatment of the classical case $K = 1$ of \textcite{Gutekunst2025}, where the Feller condition $\kappa \theta \ge \frac{\nu^2}{2}$ was required.
		\end{remark}

		\subsection{The Complete Case}
		\label{subsection:heston:complete}

        \begin{theorem}
                \label{heston:complete:verification}
                Suppose that $|\rho| = 1$.
                Denote \[
                        G^{s}_t = \E^{\Q}\left[\exp( -\eta_1 (U_s - U_t) ) \condbar \mathcal{F}_t\right] = \exp\left( \int_t^s F(\psi(s-u)) \d G_t(u) \right)
                ,\] where $ F(\psi) = -\eta_1 - \tilde{\kappa} \psi + \frac{\nu^2}{2} \psi^2 $, $\psi $ solves the Riccati-Volterra equation $\psi = F(\psi) * K$, and \[
                        G_t(s) = G_0(s) + \int_t^s \int_0^t K(u-v) \d (-\tilde{\kappa} U_v + \nu \tilde{M}_v) \d u
                .\]
                Set 
                \begin{align*}
                        F_t &= \int_t^\infty e^{-\eta_0 (s-t)} G^s_t \d s, \\
                        Z^F_t &= \nu \int_t^\infty e^{-\eta_0 (s-t)} \psi(s-t) G^s_t \d s
                .\end{align*}
                Then the value process and optimal controls are given by \eqref{verification:optimal_controls}.
        \end{theorem}
        \begin{proof}
                First, notice that since $R_\rho = 1$, $F$ is exactly the solution to \cref{hjb}.
                Moreover, $F$ is well-defined since the model is strongly myopically well-posed.
                Furthermore, $\psi$ exists and satisfies $\frac{\tilde{\kappa} - \sqrt{\tilde{\kappa}^2 + 2 \eta_1 \nu^2} }{\nu^2} < \psi < 0$ by \cite[Thm.~A.5]{Gatheral2019}.

                By \cite[Eq. (6.1),(6.7),(6.11)]{AbiJaber2021}, we have \[
				        d \log \E^{\Q}\left[ \exp(- \eta_1 U_s) \condbar \mathcal{F}_t \right] = -\frac{1}{2} \nu^2 \psi(s-t)^2 dU_t + \nu \psi(s-t) d\tilde{M}_t, \quad 0 \le t \le s
		          ,\] so the $\Q$-dynamics of $G^{s}$ are given by \[
				        dG^{s}_t = G^{s}_t (\eta_1 dU_t + \psi(s-t) \nu d\tilde{M}_t), \quad 0 \le t \le s
		          .\]

                Using the stochastic Leibniz rule on $F$, one sees that $Z^F$ is the control process of $F$.
                By the mean-value theorem, we have $ \frac{Z^F_t}{F_t} = \psi(\xi_t) \nu $ for some $\xi_t \ge 0$.
		          Since $\psi$ is bounded, $\mathcal{E}\left( \frac{Z^F}{F} \cdot \tilde{M} \right) $ is hence a $\Q$-martingale by \cite[Lem.~6.1]{AbiJaber2021}.

                Optimality now follows from \cref{verification:verification,verification:verification:sufficient_criteria}.
        \end{proof}
        
		\subsection{The Incomplete Case}
		\label{subsection:heston:incomplete}

		Now, we consider the incomplete case $|\rho| < 1$.
		Since the model is strongly myopically well-posed, the existence of an $\mathbb{F}^{M}$-adapted candidate solution follows immediately from \cref{existence:uniformly_wellposed}.
		
		When $\tilde{\kappa} \ge 0$, a better supersolution than the constant $\eta_0^{-R_{\rho}}$ can be obtained from \cref{existence:supersolution}.
        To this end, let $\tilde{\eta}_1 \in [0, \eta_1]$ be a solution to the equation \[
				\tilde{\eta}_1 = \eta_1 - \frac{1}{2} (R_{\rho}-1) \frac{1}{\nu^2} \left( \tilde{\kappa} - \sqrt{\tilde{\kappa}^2 + 2 \tilde{\eta}_1 \nu^2 } \right)^2
		.\] 
		Note that a solution $\tilde{\eta}_1 \in [0, \eta_1]$ exists by the intermediate value theorem using that $\tilde{\kappa} \ge 0$.
		Moreover, the solution is unique over $[0, \eta_1]$ by monotonicity of both sides of the equation.
		Set $\tilde{H}_t = \eta_0 t + \tilde{\eta}_1 U_t $.
		Note that $\tilde{H}$ corresponds to a model with modified excess return per unit variance \[
				\tilde{\Lambda}^2 = \frac{2R^2}{R-1} \tilde{\eta}_1 \le \Lambda^2
		.\] 
		Now, set \[
				G_t = \E^{\Q}\left[ \int_{t}^{\infty} \exp( -(\tilde{H}_s - \tilde{H}_t) ) \d s \condbar \mathcal{F}_t \right]^{R_{\rho}}
		\] and denote the control process of $G$ by $Z^G$.
		Note that $G$ is well-defined since the model associated to $\tilde{H}$ is strongly myopically well-posed as $\eta_0 > 0$ and $\tilde{\eta}_1 \ge 0$.
		From the complete case, we have $\frac{Z^G_t}{G_t} = R_{\rho} \tilde{\psi}(\tilde{\xi}_t) \nu $ for some $\tilde{\xi}_t \ge 0$, where $\tilde{\psi}$ solves \[
				\tilde{\psi} = \left( -\tilde{\eta}_1 - \tilde{\kappa} \tilde{\psi} + \frac{\nu^2}{2} \tilde{\psi}^2 \right) * K
		.\]  
		We have
		\begin{align*}
				\int_{0}^{t} R_{\rho} G_s \d H_s - \int_{0}^{t} R_{\rho} G_s \d{} \tilde{H}_s - \int_{0}^{t} \frac{1}{2} \frac{R_{\rho}-1}{R_{\rho}} \frac{(Z^{G}_s)^2}{G_s} dU_s = 
				\int_{0}^{t} R_{\rho} G_s \left(\eta_1 - \tilde{\eta}_1 - \frac{1}{2} \frac{R_{\rho}-1}{R_{\rho}^2} \frac{(Z^{G}_s)^2}{G_s^2}\right) \d U_s
		.\end{align*}
		Since $\frac{\tilde{\kappa} - \sqrt{\tilde{\kappa}^2 + 2 \tilde{\eta}_1 \nu^2} }{\nu^2} < \tilde{\psi} < 0$, we have \[
				\eta_1 - \tilde{\eta}_1 - \frac{1}{2} \frac{R_{\rho}-1}{R_{\rho}^2} \left(\frac{Z^{G}_s}{G_s}\right)^2 \ge \eta_1 - \tilde{\eta}_1 - \frac{1}{2} (R_{\rho} - 1) \nu^2 \left( \frac{\tilde{\kappa} - \sqrt{\tilde{\kappa}^2 + 2 \tilde{\eta}_1 \nu^2} }{\nu^2} \right)^2 = 0
		.\]
		Hence, $\int_{0}^{\cdot} R_{\rho} G_s \d{} \tilde{H}_s + \int_{0}^{\cdot} \frac{1}{2} \frac{R_{\rho}-1}{R_{\rho}} \frac{(Z^{G}_s)^2}{G_s} \d U_s $ is strongly majorised by $\int_{0}^{\cdot} R_{\rho} G_s \d H_s$, and so $G$ is a supersolution to \eqref{hjb} by \cref{existence:supersolution}.
		Notice that $G \le \eta_0^{-R_{\rho}}$, so the solution obtained via \cref{existence:fixed_point_iteration} with supersolution $G$ is the same as the one obtained from \cref{existence:uniformly_wellposed} by uniqueness.

        Our main result is the following:
        
        \begin{theorem}
				\label{heston:verification}
				Let $F$ be the solution to \eqref{hjb} from \cref{existence:uniformly_wellposed} with control process $Z^F$. 
                Then $Z^F \le 0$.
				Moreover, the value process and optimal controls are given by \eqref{verification:optimal_controls}.
		\end{theorem}
        \begin{proof}
                The non-positivity of $Z^F$ is given by \cref{heston:non_positive_volatility}.
				Here, we will show that $\mathcal{E}(\frac{Z^F}{F} \cdot \tilde{M})$ is a true $(\mathbb{F}^M, \Q)$-martingale. 
				This implies optimality by \cref{verification:verification,verification:verification:sufficient_criteria}.

				Fix $T > 0$, and define the sequences of stopping times $\tau^{1}_N = \inf \{ t \ge 0: U_t \ge N \} $, $\tau^2_N = \inf \{ t \ge 0: \tilde{M}_t \ge N \} $, $\tau_N = \tau^{1}_N \wedge \tau^2_N \wedge T$.
				First, we will show that $\mathcal{E}(\frac{Z^F}{F} \cdot \tilde{M})$ is a true martingale over $[0, \tau_N]$ for all $N \in \N$.
				
				To this end, fix $N \in \N$ and denote $\mu_t(s) = \E^{\Q}\left[ U_s \condbar \mathcal{F}_t \right] $.
				Set $\delta = 1$.
				Since $ \E^{\Q}[U_t^2] < \infty$ for all $t > 0$ by \cite[Lem.~3.1]{AbiJaber2021} and $U = U$, $\tilde{M}$ is a true $\Q$-martingale.
				For $t \wedge \tau_N \le s \le (t \wedge \tau_N) + \delta \le T + \delta $, we have \begin{align*}
						\mu_{t \wedge \tau_N}(s) &= 
						G_0(s) + \int_{0}^{t \wedge \tau_N} K(s-u) (-\tilde{\kappa} U_u + \nu \tilde{M}_u) \d u + \int_{t \wedge \tau_N}^{s} K(s-u) (-\tilde{\kappa} \mu_{t \wedge \tau_N}(u) + \nu \tilde{M}_{t \wedge \tau_N}) \d u \\ & \le
						G_0(s) + \int_{0}^{t \wedge \tau_N}  K(s-u) (|\tilde{\kappa}| + \nu) N \d u + \int_{t \wedge \tau_N}^{s} K(s-u) \nu N \d u \\ & \quad + |\tilde{\kappa}| \int_{t \wedge \tau_N}^{s} K(s-u) \mu_{t \wedge \tau_N}(u) \d u \\ & \le 
						G_0(T + \delta) + \int_{0}^{T+\delta} K(T+\delta-u) (|\tilde{\kappa}| + \nu) N \d u + \int_{0}^{T + \delta} K(T + \delta - u) \nu N \d u \\ & \quad + |\tilde{\kappa}| \int_{t \wedge \tau_N}^{s} K(s-u) \mu_{t \wedge \tau_N}(u) \d u \\ & =:
						C + |\tilde{\kappa}| \int_{t \wedge \tau_N}^{s} K(s-u) \mu_{t \wedge \tau_N}(u) \d u 
				.\end{align*} 
				Let $R_{|\tilde{\kappa}|}$ be the resolvent of the second kind of $|\tilde{\kappa}| K$.
				By the Grönwall lemma for convolution inequalities (see \cite[Thm.~9.8.2]{Gripenberg1990}), we have \[
						\mu_{t \wedge \tau_N}(s) \le C\left(1 + \int_{t \wedge \tau_N}^{s} R_{|\tilde{\kappa}|}(s-u) \d u \right) \le C \left(1 + \int_{0}^{T+\delta} R_{|\tilde{\kappa}|}(T + \delta - u) \d u\right) =: C'
				.\] 
				Recall that $H_t = \eta_0 t + \eta_1 U_t$.
				Hence, we have \[
						\E^{\Q}\left[ H_{s} - H_{t \wedge \tau_N} \condbar \mathcal{F}_{t \wedge \tau_N} \right] \le \E^{\Q}\left[ H_s \condbar \mathcal{F}_{t \wedge \tau_N} \right] = \eta_0 s + \eta_1 \mu_{t \wedge \tau_N}(s) \le \eta_0 (T + \delta) + \eta_1 C'
				.\]
				By condition \ref{properties:local_boudedness_away_from_0:conditional_expectation_bound} in \cref{properties:local_boundedness_away_from_0}, there exists a constant $K > 0$ s.t.\ $F \ge K$ over $[0, \tau_N]$.
				Thus, $\mathcal{E}(\frac{Z^F}{F} \cdot \tilde{M})$ is a true martingale over $[0, \tau_N]$ by \cref{properties:local_boundedness_away_from_0:martingale}.

				Define the measure $\Q^{N}$ by $\frac{\text{d}\Q^{N}}{\text{d}\Q} = \mathcal{E}(\frac{Z^F}{F} \cdot \tilde{M})_{\tau_N}$.
				Since $\mathcal{E}(\frac{Z^F}{F} \cdot \tilde{M})$ is a supermartingale, it is sufficient to prove that $ \E^{\Q}[\mathcal{E}(\frac{Z^F}{F} \cdot \tilde{M})_T] = 1$ to conclude that $\mathcal{E}(\frac{Z^F}{F} \cdot \tilde{M})$ is a true martingale over $[0, T]$.
				We have \[
						\E^{\Q}\left[\mathcal{E}\left(\frac{Z^F}{F} \cdot \tilde{M}\right)_T\right] \ge \E^{\Q}\left[\mathcal{E}\left(\frac{Z^F}{F} \cdot \tilde{M}\right)_{\tau_N} \mathbbm{1}_{\{\tau_N = T\}} \right] = \Q^{N}(\tau_N = T) = 1 - \Q^{N}(\tau_N < T)
				,\] so it remains to show that $\Q^{N}(\tau_N < T) \to 0$ as $N \to \infty$.

				By Girsanov's theorem, $U$ satisfies \[
						U_t = G_0(t) + \int_{0}^{t} K(t-s) \left(-\tilde{\kappa} U_s + \nu \int_{0}^{s} \frac{Z^F_u}{F_u} \d U_u + \nu \tilde{M}^{N}_s\right) \d s
				\] under $\Q^{N}$, where $\tilde{M}^{N} = \tilde{M} - \int_{0}^{\cdot} \frac{Z_u}{F_u} \d U_u $ is a $\Q^{N}$-local martingale with quadratic variation $U$.
				Since $Z^F \le 0$ by \cref{heston:non_positive_volatility}, we have \[
						U_t \le G_0(t) + \int_{0}^{t} K(t-s) (-\tilde{\kappa} U_s + \nu \tilde{M}^{N}_s) \d s 
				.\] 
				By \cref{volterra:a_priori_estimate}, there exists a constant $C$ that only depends on $T, G_0, \tilde{\kappa}, \nu, K$ s.t.\ \[
						\E^{\Q^{N}}\left[\sup_{t \le T} U_t^2\right] \le C
				.\] 
				Importantly, note that $C$ does not depend on $N$.

				Using Markov's inequality, we obtain \[
						\Q^{N}(\tau^{1}_N < T) \le \Q^{N}(U_T \ge N) \le \frac{1}{N^2} \E^{\Q^{N}}[U_T^2] \le \frac{C}{N^2} \to 0 \text{ as } N \to \infty
				.\] 
				Using $Z^F \le 0$, Markov's inequality, and the BDG-inequality, we obtain \begin{align*}
						\Q^{N}(\tau^{2}_N < T) &\le 
						\Q^{N}\left(\sup_{t \le T} \tilde{M}_t \ge N\right) = 
						\Q^{N}\left(\sup_{t \le T} \left(\tilde{M}^{N}_t + \int_{0}^{t} \frac{Z^F_s}{F_s} \d U_s\right) \ge N\right) \le 
						\Q^{N}(\sup_{t \le T} \tilde{M}^{N}_t \ge N) \\ & \le 
						\frac{1}{N^{4}} \E^{\Q^{N}}\left[\left(\sup_{t \le T} |\tilde{M}^{N}_t|\right)^{4}\right] \le 
						\frac{1}{N^{4}} C^{BDG}_4 \E^{\Q^{N}}[U_T^2] \le 
						\frac{C C^{BDG}_4}{N^{4}} \to 0 \text{ as } N \to \infty
				\end{align*} for some constant $C^{BDG}_4$.
				Together, we have \[
						\Q^{N}(\tau_N < T) \le \Q^{N}(\tau^{1}_N < T) + \Q^{N}(\tau^2_N < T) \to 0 \text{ as } N \to \infty
				.\qedhere\] 
		\end{proof}

        \subsection{Numerical Illustration and Comparative Statics in the Hurst Parameter}

		We now provide a numerical illustration of how the roughness of volatility affects optimal investment and consumption.
        \Cref{heston:fig:consumption} shows the sub- and supersolution for the optimal consumption rate using the fractional kernel $K_h(t) = \frac{t^{h-\frac{1}{2}}}{\Gamma\left(h+\frac{1}{2}\right)}$ and parameters from \cref{heston:table:parameters} over the range that corresponds to $10\%$--$25\%$ initial volatility.
        Note that the parameters satisfy $\tilde{\kappa} = 0.195 > 0$, so we can use the tighter supersolution from \cref{subsection:heston:incomplete}.
        The sub- and supersolutions are quite close to each other, yielding relatively tight bounds for the optimal consumption rate. 
        \Cref{heston:fig:portfolio} shows the optimal portfolio for the sub- and supersolution.
        While it is generally not clear that these yield upper and lower bounds on the optimal portfolio, they should still provide a good approximation for the optimal portfolio.
        
        Like the myopic consumption rate, the optimal consumption rate is approximately linear in the initial variance $Y_0$ over the economically relevant range. 
        The slope of the optimal consumption rate is always smaller than the slope of the myopic consumption rate. 
        This is because (unlike in the myopic case) the agent takes the expected mean reversion of the variance process into account, which smooths out the consumption.
        The slope increases as $h$ gets smaller, which means that the agent adjusts their consumption more strongly to the initial level of the variance the rougher the volatility is.

        The optimal risky asset allocation can be decomposed into the myopic portfolio $\frac{\Lambda}{R}$ and the intertemporal hedge $\frac{1}{R_\rho} \rho \frac{Z^F}{F}$.
        Similar to the myopic portfolio, the optimal risky asset allocation appears to be almost constant over the economically relevant range of the initial variance.
        Since $\rho \le 0$ and $Z^F \le 0$, the intertemporal hedging demand is always positive, so the agent invests a higher fraction of wealth into the risky asset than under the myopic strategy.
        As $h$ decreases and the volatiliy becomes more rough, the agent's intertemporal hedging demand decreases, and they invest less wealth into the risky asset.
        
		\begin{figure}[htb]
				\centering

				\begin{minipage}[t]{.45\textwidth}
						\includegraphics[width=0.95\textwidth]{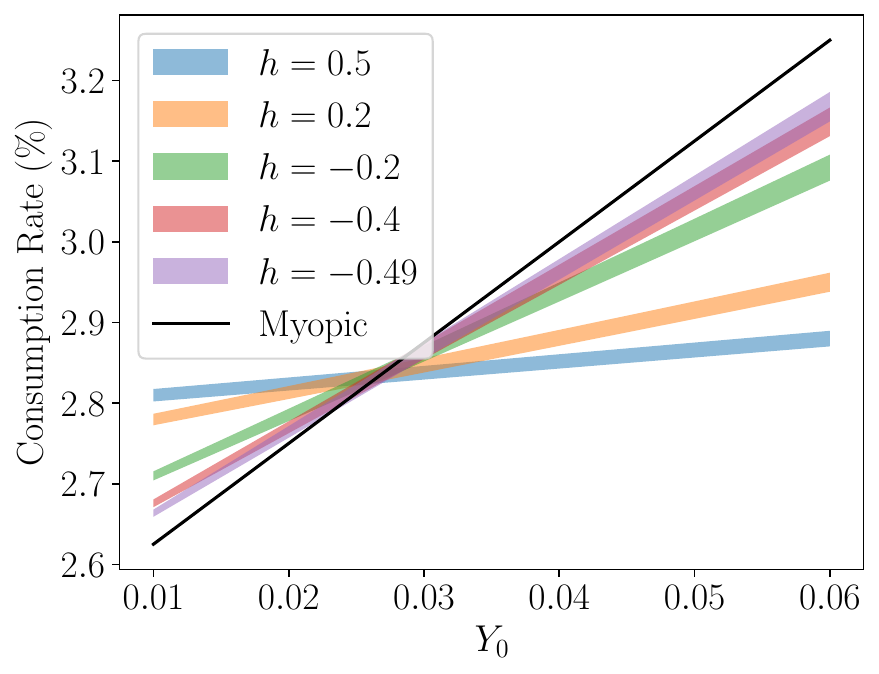}
						\caption{Upper and lower bounds on the optimal consumption rate at $t = 0$. The slope increases as $h$ decreases.}
						\label{heston:fig:consumption}		
				\end{minipage}%
                \hspace{0.05\textwidth} %
				\begin{minipage}[t]{.45\textwidth}
						\includegraphics[width=0.95\textwidth]{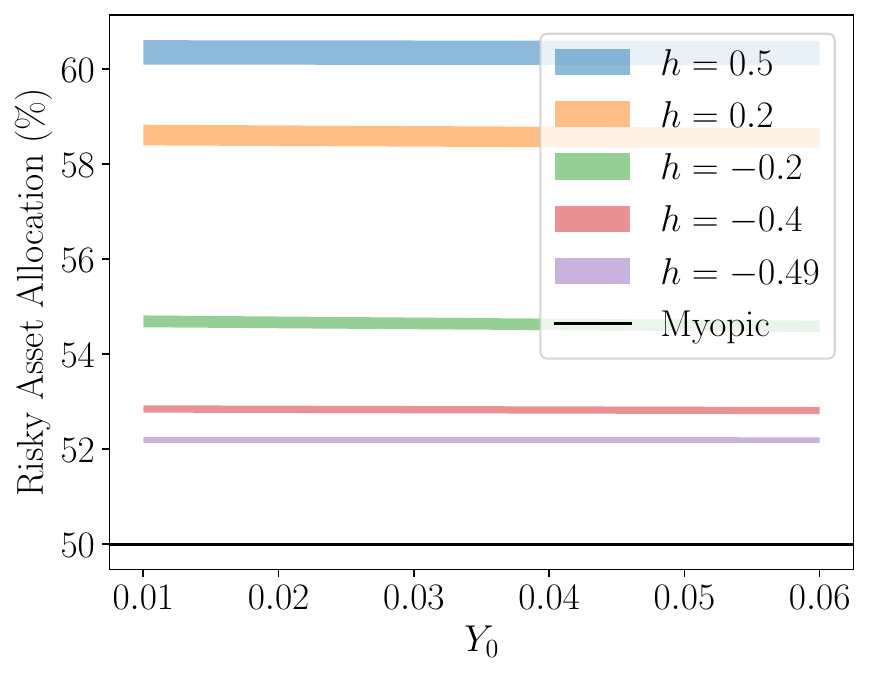}
						\caption{Approximate optimal portfolio at $t = 0$. The allocation decreases as $h$ decreases.}
						\label{heston:fig:portfolio}		
				\end{minipage}
		\end{figure}

        \begin{table}[htb]
				\centering
				\begin{tabular}{cccccccc}
						\toprule
						$R$   & $\delta$ & $r$    & $\Lambda$ & $\kappa$ & $\theta$ & $\nu$ & $\rho$ \\
						\midrule
						$2$ & $0.03$   & $0.02$ & $1$    & $0.3$    & $0.02$    & $0.3$ & $-0.7$ \\
						\bottomrule
				\end{tabular}
				\caption{Parameters of the Volterra Heston model, adapted from \cite[Table~4]{AbiJaber2019a}.}
				\label{heston:table:parameters}
		\end{table}

        \appendix

        \section{Auxiliary Results and Proofs}
        \label{appendix:auxiliary}

        \begin{proof}[Proof of \cref{existence:dynamics}]
            First, let $F$ be a solution to \eqref{hjb}.
            Set $ N_t = \int_{0}^{t} D_s R_{\rho} F_s^{1-\frac{1}{R_{\rho}}} \d s + D_t F_t $.
            Since $F$ is a solution to \eqref{hjb}, we have \[
                    N_t = \int_{0}^{t} D_s R_{\rho} F_s^{1-\frac{1}{R_{\rho}}} \d s + D_t \E^{\Q}\left[\int_{t}^{\infty} \frac{D_s}{D_t} R_{\rho} F_s^{1-\frac{1}{R_{\rho}}} \condbar \mathcal{F}_t\right] = \E^{\Q}\left[\int_{0}^{\infty} D_s R_{\rho} F_s^{1-\frac{1}{R_{\rho}}} \d s \condbar \mathcal{F}_t \right] 
            ,\] so $N$ is a (uniformly integrable) $(\mathbb{F}, \Q)$-martingale.
            By \cref{setting:assumption:M-adapted}, $D$ (and hence the entire integrand) is $\mathbb{F}^{M}$-adapted.
            Using that $\mathbb{F}^{M}$ is immersed in $\mathbb{F}$ by \cref{setting:assumption:M-PRP}, we thus have \[
                    N_t = \E^{\Q}\left[ \int_{0}^{\infty} D_s R_{\rho} F_s^{1-\frac{1}{R_{\rho}}} \d s \condbar \mathcal{F}^{M}_t \right] 
            ,\] so $N$ is also an $(\mathbb{F}^{M}, \Q)$-martingale.
              By \cref{setting:assumption:M-PRP}, there exists an $\mathbb{F}^M$-predictable process $Z^F$ s.t. \ \[
                    dN_t = D_t R_{\rho} F_t^{1-\frac{1}{R_{\rho}}} dt - D_t R_{\rho} F_t dH_t + D_t dF_t = D_t Z^F_t d\tilde{M}_t
            ,\] i.e. \[
                    dF_t = -R_{\rho} F^{1-\frac{1}{R_{\rho}}}_t dt + R_{\rho} F_t dH_t + Z^F_t d\tilde{M}_t
            .\] 
            Moreover, $\int_{0}^{\cdot } D_s Z^F_s \d{} \tilde{M}_s = N - N_0$ is a true martingale.
            Finally, we have \[
                    \E^{\Q}\left[D_T F_T \condbar \mathcal{F}_t \right] = \E^{\Q}\left[N_T - \int_{0}^{T} D_s R_{\rho} F_s^{1-\frac{1}{R_{\rho}}} \d s \condbar \mathcal{F}_t \right] = \E^{\Q}\left[\int_{T}^{\infty} D_s R_{\rho} F_s^{1-\frac{1}{R_{\rho}}} \d s \condbar \mathcal{F}_t \right] \to 0 \text{ as } T \to \infty
            \] by dominated convergence.

            Now, let $F, Z$ be processes that satisfy \ref{existence:dynamics:drift}-\ref{existence:dynamics:transversality}, and define $N$ as above.
            By \ref{existence:dynamics:drift} and \ref{existence:dynamics:martingale}, $N$ is a true martingale.
            Hence, we have \[
                    D_t F_t = \E^{\Q}\left[ D_T F_T + \int_{t}^{T} D_s R_{\rho} F_s^{1-\frac{1}{R_{\rho}}} \d s \condbar \mathcal{F}_t \right]
            .\] 
            Taking limits along an approximating sequence in \ref{existence:dynamics:transversality} and using monotone convergence, we obtain \[
                    D_t F_t = \E^{\Q}\left[ \int_{t}^{\infty} D_s R_{\rho} F_s^{1-\frac{1}{R_{\rho}}} \condbar \mathcal{F}_t \right]
            ,\] so $F$ is a solution to \eqref{hjb}.
        \end{proof}

        Next, we show that in order for two solutions to be ordered, it is sufficient for them to be ordered up to a constant factor.
		This will be useful for showing uniqueness of the solution.

		\begin{lemma}
				\label{existence:order}
				Assume that $R > 1$ and let $F, G \in \mathcal{D}(T)$ be solutions to \eqref{hjb}. 
				Suppose that either of the following conditions is satisfied:
				\begin{enumerate}[(a)]
						\item \label{existence:order:linear}
								$F \le CG$ for some $C > 0$, or
						\item \label{existence:order:power}
								$F \le CG^{r}$ for some $C > 0$, $r \in (0, 1)$, and $T^{n} \one \to F$.
				\end{enumerate}
				Then $F \le G$.
		\end{lemma}
		\begin{proof}
				If $R_{\rho} = 1$, \eqref{hjb} has a unique solution, so $F = G$. 
				Assume now that $R_{\rho} > 1$.

				For Case \ref{existence:order:linear}, notice that by \cref{existence:operator:properties} we have \[
						F = TF \le T (CG) = C^{1-\frac{1}{R_{\rho}}} TG = C^{1-\frac{1}{R_{\rho}}} G
				.\] 
				Inductively, this yields that $F \le C^{\left( 1-\frac{1}{R_{\rho}} \right)^{n} } G$ for all $n \in \N$.
				Since $R_{\rho} > 1$, $(1 - \frac{1}{R_{\rho}})^{n} \to 0$ as $n \to \infty$, and so we obtain $F \le G$.

				In case \ref{existence:order:power}, the Hölder inequality of \cref{existence:operator:properties} with $p=\frac{1}{1-r}$ and $q = \frac{1}{r}$ yields \[
						F = TF \le T(CG^{r}) = C^{1-\frac{1}{R_{\rho}}} TG^{r} \le C^{1-\frac{1}{R_{\rho}}} (T \one)^{1-r} (TG)^{r} = C^{1-\frac{1}{R_{\rho}}} (T \one)^{1-r} G^{r}
				.\] 
				Iterating this yields \[
						F \le C^{(1-\frac{1}{R_{\rho}})^{n}} (T^{n} \one)^{1-r} G^{r} 
				\] for all $n \in \N$.
				Taking the limit $n \to \infty$, we obtain $F \le F^{1-r} G^{r}$, and thus $F \le G$.
		\end{proof}

        The next lemma collects some conditions under which the local boundedness away from zero required in \cref{properties:local_boundedness_away_from_0:martingale} can be verified in the complete market case (which is a subsolution in the incomplete case).

        \begin{lemma}
				\label{properties:local_boundedness_away_from_0}
				Assume that \[
						F_t = \E^{\Q}\left[ \int_{t}^{\infty} \exp(-(H_s - H_t)) \d s \condbar \mathcal{F}_t \right] 	 
				\] is well-defined.
				Let $\tau$ be an $\mathbb{F}$-stopping time, and assume that one of the following conditions is satisfied:
				\begin{enumerate}[(a)]
						\item \label{properties:local_boundedness_away_from_0:norm_bound} There exist constants $C, \delta, \varepsilon > 0$ s.t.\ $\Q\left(\sup_{t \wedge \tau \le s \le (t \wedge \tau) + \delta} H_s - H_{t \wedge \tau} \le C \condbar \mathcal{F}_{t \wedge \tau}\right) \ge \varepsilon$ a.s.
						\item \label{properties:local_boundedness_away_from_0:stopping_time} There exist an $\mathbb{F}$-stopping time $\sigma \ge \tau$ and constants $C, \delta, \varepsilon > 0$ s.t.\ $\Q\left(\sigma - \tau \ge \delta \condbar \mathcal{F}_{t \wedge \tau}\right) \ge \varepsilon$ a.s.\ and $\sup_{t \wedge \tau \le s \le ((t \wedge \tau) + \delta) \wedge \sigma} H_s - H_{t \wedge \tau} \le C$ a.s.
						\item \label{properties:local_boudedness_away_from_0:conditional_expectation_bound} There exist constants $C, \delta > 0$ s.t.\ $ \E^{\Q}\left[ H_s - H_{t \wedge \tau} \condbar \mathcal{F}_{t \wedge \tau} \right] \le C$ a.s.\ for $t \wedge \tau \le s \le (t \wedge \tau) + \delta$.
				\end{enumerate}	
				Then there exists a constant $K > 0$ s.t. $F_{t \wedge \tau} \ge K$ for all $t \ge 0$.
		\end{lemma}
		\begin{proof}
				Fix $t \ge 0$.
				In case \ref{properties:local_boundedness_away_from_0:norm_bound}, we have 
				\begin{align*}
						F_{t \wedge \tau} &= 
						\E^{\Q}\left[ \int_{t \wedge \tau}^{\infty} \exp( -(H_s - H_{t \wedge \tau})) \d s \condbar \mathcal{F}_{t \wedge \tau} \right] \\ & \ge 
						\E^{\Q}\left[ \int_{t \wedge \tau}^{(t \wedge \tau) + \delta} \exp(-(H_s - H_{t \wedge \tau})) \d s \condbar \mathcal{F}_{t \wedge \tau} \right] \\& \ge 
						\E^{\Q}\left[ \mathbbm{1}_{\{\sup_{t \wedge \tau \le s \le (t \wedge \tau) + \delta} H_s - H_{t \wedge \tau} \le C\}} \int_{t \wedge \tau}^{(t \wedge \tau) + \delta} \exp(-(H_s - H_{t \wedge \tau})) \d s \condbar \mathcal{F}_{t \wedge \tau} \right] \\& \ge 
						\delta e^{-C} \Q\left( \sup_{t \wedge \tau \le s \le (t \wedge \tau) + \delta} H_s - H_{t \wedge \tau} \le C \condbar \mathcal{F}_{t \wedge \tau} \right) \ge 
						\varepsilon \delta e^{-C}
				.\end{align*}
				Similarly, in case \ref{properties:local_boundedness_away_from_0:stopping_time}, we have \[
						F_{t \wedge \tau} \ge \E^{\Q}\left[ \mathbbm{1}_{\{\sigma - \tau \ge \delta\}} \int_{t \wedge \tau}^{((t \wedge \tau) + \delta) \wedge \sigma} \exp(-(H_s - H_{t \wedge \tau})) \d s \condbar \mathcal{F}_{t \wedge \tau} \right] \ge \varepsilon \delta e^{-C}
				.\]	
				Finally, in case \ref{properties:local_boudedness_away_from_0:conditional_expectation_bound}, applying Fubini-Tonelli and Jensen's inequality yields \[
						F_{t \wedge \tau} \ge \int_{t \wedge \tau}^{(t \wedge \tau) + \delta} \exp\left( - \E^{\Q}\left[ H_s - H_{t \wedge \tau} \condbar \mathcal{F}_{t \wedge \tau} \right] \right)  \d s \ge \delta e^{-C}
				.\qedhere\] 
		\end{proof}

        \begin{remark}
				\label{properties:local_boundedness_away_from_0:examples}
				Condition \ref{properties:local_boundedness_away_from_0:norm_bound} is satisfied if $H$ (or alternatively $Y$ in the absolutely continuous setting of \cref{setting:typical_factor_model}) has sufficient path-regularity, e.g.\ Hölder-continuous with a.s.\ bounded Hölder-constant.

				Condition \ref{properties:local_boundedness_away_from_0:stopping_time} is fulfilled in Markovian models:
				Assume that we are in the absolutely continuous setting of \cref{setting:typical_factor_model} with $\eta \in L^{\infty}_{loc}$, and that $Y$ is a time-homogeneous strong Markov process with continuous sample paths. 
				Fix $T > 0$ and $N \in \N$, and set $\tau = \inf \{t: |Y_t| \ge N\} \wedge T $ and $\sigma = \inf \{t: |Y_t| \ge N+1\} \wedge (T+1)$.
				Since $\Q\left(\sigma - \tau \ge \delta \condbar \mathcal{F}_{t \wedge \tau}\right) \ge \Q_N(\sigma \ge \delta) \wedge \Q_{-N}(\sigma \ge \delta)$ (where $\Q_x$ is the measure under which $Y_0 = x$), we can for any $\varepsilon \in (0, 1)$ find a $\delta > 0$ with $\Q\left(\sigma - \tau \ge \delta \condbar \mathcal{F}_{t \wedge \tau}\right) \ge \varepsilon$ a.s. 
				Clearly, we also have $\tau \le \sigma$ and $\sup_{t \wedge \tau \le s \le ((t \wedge \tau) + \delta) \wedge \sigma} H_s - H_{t \wedge \tau} \le \delta \|\eta\|_{L^{\infty}([-N-1, N+1])}$.

				Condition \ref{properties:local_boudedness_away_from_0:conditional_expectation_bound} is used in the Volterra Heston model, see \cref{section:heston}.
		\end{remark}
        
        Using an approximation argument, one can use \cref{properties:volatility:non-positive} to show that the $Z$-component of the BSDE is non-positive in the Volterra Heston model.

		\begin{proposition}
				\label{heston:non_positive_volatility}
				Let $F$ be the solution to \eqref{hjb} from \cref{existence:uniformly_wellposed} with control process $Z^F$.
				Then $Z^F$ satisfies $Z^F \le 0$.
		\end{proposition}
		\begin{proof}
				\textbf{Step 1:}
				Assume that $K$ is non-singular (i.e.\ $K(0) < \infty$) and locally Lipschitz.
				Since $K$ is non-singular, we have $K \in L^2_{\mathrm{loc}}$.
				Hence, we can work with the spot variance process in the setting of \cref{setting:typical_factor_model} instead of with the integrated processes.
				We will approximate the Volterra Heston model by a sequence of regularised models.
				To this end, let $b(y) = -\tilde{\kappa} (y - \tilde{\theta})$, let $\sigma^{\varepsilon}(y) = \nu \frac{\sqrt{y+\varepsilon} - \sqrt{\varepsilon}}{1 + \varepsilon (\sqrt{y+\varepsilon} - \sqrt{\varepsilon})}$, and set $\eta^{\varepsilon}(y) = \eta_0 + \eta_1 (y \wedge \varepsilon^{-1})$.
                Then $\sigma^\varepsilon$ is globally Lipschitz, bounded, and satisfies $\sigma^\varepsilon \to \nu \sqrt{\cdot} $ uniformly over compact sets as $\varepsilon \to 0$.
				Consider the regularised factor process (under $\Q$) \[
						Y^{\varepsilon}_t = Y^{\varepsilon}_0 + \int_{0}^{t} K(t-s) b(Y^{\varepsilon}_s) \d s + \int_{0}^{t} K(t-s) \sigma^{\varepsilon}(Y^{\varepsilon}_s) \d B^{\Q}_s
				\] for some $\Q$-Brownian motion $W^{\Q}$.
				Since the coefficients $b, \sigma^{\varepsilon}$ are Lipschitz, this equation has a unique non-negative strong solution by \cite[Thm.~3.3,3.6]{AbiJaber2019}.
				Since the $Y^{\varepsilon}$ are strong solutions, we may assume that they are defined on a common probability space w.r.t.\ a common driving Brownian motion. 
				Moreover, the $Y^{\varepsilon}$ are adapted to the Brownian filtration.

				We consider the model with $\tilde{M}^{\varepsilon} = \int_{0}^{\cdot} \frac{1}{\nu} \sigma^{\varepsilon}(Y^{\varepsilon}_s) \d W^{\Q}_s$ and $H^{\varepsilon} = \int_{0}^{\cdot} \eta^{\varepsilon}(Y^{\varepsilon}_s) \d s$.
				Denote the solution to \eqref{hjb} in this model by $F^{\varepsilon}$ with control process $Z^{\varepsilon}$.
				Notice that $H^{\varepsilon}$ is $(\eta_0 + \eta_1 \varepsilon^{-1})$-Lipschitz.
				Hence, $F^\varepsilon$ is bounded away from $0$ (see \cref{existence:eta_bounded}).
				By \cite[Thm.~4.3]{Coutin2001}, $Y^{\varepsilon}$ is Malliavin-differentiable, and for $u \le t$ we have \[
						D_u Y^{\varepsilon}_t = K(t-u) \sigma^{\varepsilon}(Y^{\varepsilon}_u) + \int_{u}^{t} K(t-s) (-\tilde{\kappa} D_u Y^{\varepsilon}_s) \d s + \int_{u}^{t} K(t-s) (\sigma^\varepsilon)'(Y^\varepsilon_s) D_u Y^{\varepsilon}_s \d B^{\Q}_s
				.\] 
                We now show that $D_uY^\varepsilon$ is nonnegative using a straightforward adaption of invariance results for stochastic Volterra equations on $\mathbb{R}_+$ as in \cite[Theorem A.2]{abi2019markovian} . Fix $u \geq 0$. The process $(D_uY^\varepsilon_t)_{t\ge u}$ satisfies the stochastic Volterra equation
\begin{align}\label{eq:Xvolt}
    X_t = g_u(t) +\int_u^t K(t-s)b(X_s)\mathrm ds +\int_u^t K(t-s)\sigma(\omega,X_s)\mathrm dW_s^{\mathbb Q},
\end{align}
where
$$
g_u(t)=K(t-u)\sigma^\varepsilon(Y_u^\varepsilon),
\qquad b(x)=-\tilde{\kappa}x,
\qquad
\sigma(\omega,x)
=(\sigma^\varepsilon)'(Y^\varepsilon_s(\omega)) x.
$$
Observe that both coefficients vanish at the boundary $x=0$, i.e.~$b(0)=0, \sigma(\omega,0)=0$
almost surely. Moreover, the $\mathcal F_u$-measurable input curve $g_u$ is admissible in the sense of \cite[Eq.~(2.5)]{abi2019markovian}. Indeed,
$
K(\cdot-u)=K*\delta_u
$, with $\delta_u$ the Dirac mass at $u$, so that the computations of \cite[Example~2.2(ii)]{abi2019markovian} extend immediately to this setting. Therefore, the proof of \cite[Thm.~A.2]{abi2019markovian} carries over, with only minor modifications to accommodate the random input curve $g_u$ and the random coefficient $\sigma$, yielding the existence of a solution  $X$ to \eqref{eq:Xvolt} that remains nonnegative. Since the coefficients $b$ and $\sigma$ are Lipschitz continuous in $x$ (with a deterministic Lipschitz constant since the random coefficient in $\sigma$ is bounded), pathwise uniqueness holds for \eqref{eq:Xvolt}. Consequently, by strong uniqueness, the solution $(D_uY_t^\varepsilon)_{t\ge u}$ coincides with the nonnegative solution, and therefore
$ D_uY_t^\varepsilon \ge 0$, for all $t\ge u,$
almost surely.

                Moreover, adapting the proof of \cite[Lem.~3.1]{AbiJaber2019} (resp.\ \cite[Thm.~A.1]{abi2019markovian}) to accomodate a random input curve yields that \[
                        \sup_{0 \le u \le t \le T} \E^{\Q}[|D_u Y^\varepsilon_t|^2] < \infty
                \] for all $T > 0$ since $\sup_{0 \le u \le t \le T} |g_u(t)| < \infty$.
                Hence, we have $\E^\Q[\int_0^T \int_u^T |D_u \eta^\varepsilon(Y^\varepsilon_t)|^2 \d t \d u] < \infty$ and $D_u \eta^\varepsilon(Y^\varepsilon_t) \ge 0$ by the Malliavin chain rule (see \cite[Prop.~1.2.4]{Nualart2006}).
                
                Thus, we obtain $Z^{\varepsilon} \le 0$ by \cref{properties:volatility:non-positive}.

				Notice that the coefficients $b, \sigma^{\varepsilon}$ satisfy a common linear growth condition.
				Hence, the family $(Y^{\varepsilon})_{\varepsilon \in (0,1)}$ is tight over $C$ by \cite[Lem.~A.1]{AbiJaber2019}.
		By \cite[Lem.~A.2]{AbiJaber2019}, there exists a subsequence $(\varepsilon_n)_{n \in \N}$ with $\varepsilon_n \to 0$ and $(Y^{\varepsilon_n}, \tilde{M}^{\varepsilon_n}) \to (Y, \tilde{M})$ weakly over $C \times C$ for processes $(Y, M)$ that satisfy \[
						Y_t = Y_0 + \int_{0}^{t} K(t-s) b(Y_s) + \int_{0}^{t} K(t-s) \nu \sqrt{Y_s}  \d B^{\Q}_s	
				\] and $\tilde{M} = \int \sqrt{Y_s} \d B^{\Q}_s $ for some Brownian motion $B^{\Q}$, i.e.\ the limiting model is exactly the Volterra Heston model.

				It remains to verify the conditions of \cref{stability:uniformly_wellposed}.
		Note that $(H^{\varepsilon_n}, \tilde{M}^{\varepsilon_n}) \to (H, \tilde{M})$ weakly over $C \times C$ by the continuous mapping theorem.
				Moreover, notice that $H^{\varepsilon_n}_s - H^{\varepsilon_n}_t \ge \eta_0 (s-t)$ for all $n \in \N$.
				For all $t \ge 0$, there exists a constant $C_t$ (independent of $n$) s.t.\ \[
						\sup_{s \le t} \E^{\Q}[(Y^{\varepsilon_n}_s)^2] \le C_t 
				\] by \cite[Lem.~3.1]{AbiJaber2019}.
				Hence, \[
						\sup_n \E^{\Q}[H^{\varepsilon_n}_t] = \sup_n \int_{0}^{t} \E^{\Q}[\eta^{\varepsilon_n}(Y^{\varepsilon_n}_s)] \d s \le \eta_0 t + \eta_1 t \sup_n \sup_{s \le t} \E^{\Q}[Y^{\varepsilon_n}_s] \le \eta_0 t + \eta_1 t \sqrt{C_t}  < \infty
				\] for all $t \ge 0$.
				Thus, all conditions of \cref{stability:uniformly_wellposed} are satisfied, and so $\int_{0}^{\cdot} Z^{\varepsilon_n}_s \d \langle \tilde{M}^{\varepsilon_n} \rangle_s \to \int_{0}^{\cdot} Z^F_s \d U_s$ weakly over $C$.
				Since $Z^{\varepsilon_n} \le 0$ for all $n \in \N$, this implies that $Z^F \le 0$.

				\textbf{Step 2:}
				Now, we consider a general kernel $K$. 
				We will approximate $K$ by the shifted kernel $\Delta_{\varepsilon} K$, $\varepsilon \in (0, 1)$.
				Since $K$ is non-increasing and continuously differentiable on $(0, \infty)$, $\Delta_{\varepsilon} K$ is non-singular and locally Lipschitz for all $\varepsilon > 0$.
				Moreover, $\Delta_{\varepsilon} K \to K$ in $L^{1}_{loc}$.

				Let $(U^{\varepsilon}, \tilde{M}^{\varepsilon}, \mathbb{F}^{\varepsilon})$ be a weak solution to \[
						U^{\varepsilon}_t = G_0(t) + \int_{0}^{t} \Delta_\varepsilon K(t-s) (-\tilde{\kappa} U^{\varepsilon}_s + \nu \tilde{M}^{\varepsilon}_s) \d s
				\] under a measure $\Q^{\varepsilon}$, where $U^{\varepsilon}$ is continuous, non-negative, and non-decreasing, $\tilde{M}^{\varepsilon}$ is a continuous local martingale with $\langle \tilde{M}^{\varepsilon} \rangle = U^{\varepsilon}$.
				By the proof of \cite[Thm.~2.8]{AbiJaber2021}, there exists a subsequence $\varepsilon_n$ with $\varepsilon_n \to 0$ s.t.\ $(U^{\varepsilon_n}, \tilde{M}^{\varepsilon_n}) \to (U, \tilde{M})$ weakly over $C \times C$ as $n \to \infty$, where $U$ is a continuous non-decreasing solution to \[
						U_t = G_0(t) + \int_{0}^{t} K(t-s) (-\tilde{\kappa} U_s + \nu \tilde{M}_s) \d s
				\] and $\tilde{M}$ is a continuous local martingale with $U = U$.

				Set $H^{\varepsilon}_t = \eta_0 t + \eta_1 U^{\varepsilon}_t$.
				Let $F^{\varepsilon}$ be the solution to \eqref{hjb} in the model with stochastic factor $(H^{\varepsilon}, \tilde{M}^{\varepsilon})$ and control process $Z^{\varepsilon}$.

				As above, it remains to verify the conditions of \cref{stability:uniformly_wellposed}.
				Since $H$ is affine in $U$, we have $(H^{\varepsilon_n}, \tilde{M}^{\varepsilon_n}) \to (H, \tilde{M})$ weakly over $C \times C$.
				We clearly have $H^{\varepsilon_n}_s - H^{\varepsilon_n}_t \ge \eta_0 (s-t)$ for all $0 \le t \le s$.
				By \cite[Lem.~3.1]{AbiJaber2021}, there exists a constant $C_t$ (independent of $\varepsilon$) s.t.\ \[
						\E^{\Q^{\varepsilon_n}}[U_t^2] \le C_t (1 + G_0(t)^2 + \|\Delta_\varepsilon K\|_{L^{1}([0, t])}) (1 + \|R^K_{\varepsilon}\|_{L^{1}([0, t])}) 
				\] for all $t \ge 0$, where $R^K_{\varepsilon}$ is the resolvent of the second kind\footnote{The resolvent of the second kind $R^K \in L^{1}_{loc}$ is the unique solution to $R^K * K = K * R^K = R^K - K$. The resolvent of the second kind exists for any kernel $K \in L^{1}_{loc}$, see \cite[Thm.~2.3.1,2.3.5]{Gripenberg1990}. If $K \ge 0$, then $R^K \ge 0$ as well. Note that \cite{Gripenberg1990} uses the alternative convention $R^K * K = K * R^K = K - R^K$.} of the kernel $C_t \|\Delta_\varepsilon K\|_{L^{1}([0, t])} \Delta_\varepsilon K$. 
				Since the dependence of the resolvent of the second kind on the kernel is $L^{1}_{loc}$-continuous (see \cite[Thm.~2.3.1]{Gripenberg1990}) and $\Delta_{\varepsilon_n} K \to K$ in $L^{1}_{loc}$, we have $\sup_n \{ \|\Delta_{\varepsilon_n} K\|_{L^{1}([0, t])} + \|R^K_{\varepsilon_n}\|_{L^{1}([0, t])}\} < \infty$, and hence \[
						\sup_n \E^{\Q^{\varepsilon_n}}[H^{\varepsilon_n}_t] = \eta_0 t + \eta_1 \sup_n \E^{\Q^{\varepsilon_n}}[U^{\varepsilon_n}_t] < \infty
				\] for all $t \ge 0$.
				Thus, \cref{stability:uniformly_wellposed} yields $\int_{0}^{\cdot} Z^{\varepsilon_n}_s \d \langle \tilde{M}^{\varepsilon_n} \rangle_s \to \int_{0}^{\cdot} Z^F_s \d U_s$ weakly over $C$.
				Since $Z^{\varepsilon_n} \le 0$ for all $n \in \N$ by Step 1, we also have $Z^F \le 0$.
		\end{proof}
        
        The next result is a straightforward extension of the a priori estimate for the solution to an affine stochastic Volterra equation with $L^1_{\mathrm{loc}}$-kernel from \cite[Lem.~3.1]{AbiJaber2021} to subsolutions.
        
		\begin{lemma}
				\label{volterra:a_priori_estimate}

				Fix $K \in L^{1}_{loc}(\R_+, \R)$ and $G_0$ locally bounded. 
				Assume that there exists a non-decreasing non-negative and adapted process $\tilde{X}$ satisfying \[
						\tilde{X}_t \le G_0(t) + \int_{0}^{t} K(t-s) \tilde{Z}_s \d s 
				,\] where $\tilde{Z}$ is a semimartingale with characteristics $(\tilde{B}, \tilde{C}, \tilde{\nu})$ s.t.\ \[
				|\tilde{B}_t| + |\tilde{C}_t| + \int_{[0,t] \times \R} \zeta^2 \tilde{\nu}_s(\d \zeta) \le \kappa_L \tilde{X}_t \text{ a.s.} 
				\] for some constant $\kappa_L$ and all $t \ge 0$.
				Then, for all $T > 0$, \[
						\E[\sup_{t \le T} \tilde{X}_t^2] \le C(T, \kappa_L) (1 + \sup_{t \le T} |G_0(t)|^2 + \|K\|_{L^{1}([0, T])}^2) (1 + \|R^K\|_{L^{1}([0, T])}
				,\] where $C(T, \kappa_L) > 0$ depends only on $(T, \kappa_L)$, and $R^K$ is the resolvent of the second kind of the kernel $C(T, \kappa_L) \|K\|_{L^{1}([0, T])} |K|$.
		\end{lemma}
		\begin{proof}
				Notice that the proof of \cite[Lem.~3.1]{AbiJaber2021} only depends on the relation $\tilde{X}_t = G_0(t) + \int_{0}^{t} K(t-s) \tilde{Z}_s \d s $ through establishing that \[
						|\tilde{X}_t| \mathbbm{1}_{\{t < \tau_n\}} \le |G_0(t)| + \left| \int_{0}^{t} K(t-s) \tilde{Z}_s \mathbbm{1}_{\{s < \tau_n\}} \d s \right|
				,\] where $\tau_n$ is some stopping time.
				Since $\tilde{X}$ is non-negative, the latter also holds true if we only demand $\tilde{X}_t \le G_0(t) + \int_{0}^{t} K(t-s) \tilde{Z}_s \d s $.
				The rest of the proof goes through unchanged.
		\end{proof}

		\printbibliography
\end{document}